\documentclass[]{interact}

\usepackage{epstopdf}
\usepackage[caption=false]{subfig}
\usepackage{bigints}
\DeclareMathOperator*{\spann}{span}

\usepackage{dsfont}
\usepackage{multirow}
\usepackage[ruled,norelsize]{algorithm2e}
\usepackage{float}

\usepackage[natbibapa,nodoi]{apacite}
\theoremstyle{plain}
\newtheorem{theorem}{Theorem}[section]

\newtheorem{corollary}[theorem]{Corollary}
\newtheorem{proposition}[theorem]{Proposition}

\theoremstyle{definition}

\theoremstyle{remark}
\newtheorem{remark}{Remark}

\begin{document}

\articletype{ }

\title{Near Optimal Interpolation based Time-Limited Model Order Reduction }

\author{
\name{Kasturi Das\textsuperscript{a}\thanks{CONTACT: Kasturi Das. Email: kasturidas@iitg.ac.in}, Srinivasan Krishnaswamy\textsuperscript{a} and Somanath Majhi\textsuperscript{a}}
\affil{\textsuperscript{a}Indian Institute of Technology Guwahati, India}
}

\maketitle

\begin{abstract}
This paper presents an interpolatory framework for time-limited $H_2$ optimal model order reduction named Limited Time Iterative Rational Krylov Algorithm (LT-IRKA). The algorithm yields high fidelity reduced order models over limited time intervals of the form, $\begin{bmatrix}0 & \tau \end{bmatrix}$ with $\tau < \infty$ for linear time invariant (LTI) systems. Using the time limited $H_2$ norm, we derive interpolation based $H_{2,\tau}$ optimality conditions. The LT-IRKA yields a near optimal $H_2(\tau)$ reduced order system. The nearness to the exact $H_2(\tau)$ optimal reduced system is quantized in terms of the errors in the interpolation based $H_2(\tau)$ optimality conditions. We demonstrate with numerical examples how the proposed algorithm nearly satisfies the time-limited optimality conditions and also how it performs with respect to the Time-Limited Two sided Iteration Algorithm (TL-TSIA), the Time-Limited Balanced Truncation (TL-BT), the Iterative Rational Krylov Algorithm (IRKA) and the Time-Limited Pseudo Optimal Rational Krylov (TL-PORK) Algorithm over a finite time interval.
\end{abstract}




\section{Introduction}

Linear dynamical models are used to model the behaviour of physical systems. Large scale models capture the dynamics of the system to a high degree of accuracy. Complex systems require large models to capture the system dynamics accurately. However, we need a lot of computational resources to simulate or analyse such large models. Designing controllers for such large scale systems also becomes a difficult task. Using model order reduction techniques, one can resolve such issues where large models are replaced by smaller ones based on various performance measures.

A wide range of model reduction algorithms is available in the literature. Singular Value Decomposition (SVD) based model reduction methods include Balanced Truncation (BT) \citep{moore1981principal, mullis1976synthesis}, Optimal Hankel Norm Approximation (OHNA) \citep{glover1984all}, Balanced Singular Perturbation Approximation (BSPA) \citep{liu1989singular} etc. Preservation of stability and existence of an apriori error bound are some advantages of BT. The algorithm involves solving two large scale Lyapunov equations and is computationally expensive. Several modifications are incorporated in BT using Rational Krylov Subspace methods \citep{gugercin2003modified,sorensen2002sylvester,penzl2006algorithms,li2002low,gugercin2004survey} to reduce the computational cost. In OHNA, the Hankel norm error is minimized to obtain reduced-order models. BSPA introduces certain modifications in BT, resulting in reduced-order models that approximate the original system's behaviour very well at very low frequencies. 

In moment matching based model reduction techniques, the transfer function of the reduced model interpolates the original transfer function and a finite number of its moments at selected interpolation points. Algorithms based on moment matching are discussed in \citet{astolfi2010model}, \citet{feng2015fully}. They also include Krylov based reduction strategies such as the Lancos algorithm \citep{lanczos1950iteration, feldmann1997interconnect}, the Arnoldi algorithm \citep{arnoldi1951principle} and Rational Krylov based algorithms \citep{villemagne1987model,gallivan1996rational,grimme1997krylov}. The Krylov algorithms require minimal computational effort and hence work very well for large scale systems. 

\citet{gugercin2008h_2} prove that $H_2$ error norm is greater than output error norm at every time instant for all inputs having bounded energy. This result inspires a popular category of reduction methods based on optimizing the $H_2$ error norm. The $H_2$ optimal model reduction problem is a non-linear non-convex optimization problem. Finding global minimizers for such optimization problems is a difficult task. Hence, the existing methods focus on finding local minimizers. Such methods are of two categories: optimization-based methods and tangential interpolation methods.

Among the optimization methods, the earliest ones obtain the gradients of the $H_2$ error norm concerning the reduced state matrices and applied gradient flow techniques to get the local approximation of the model \citep{wilson1970optimum}. Later techniques convert the problem to an unconstrained optimization problem over the Stiefel manifold. Then the cost function is minimized over the manifold using optimization schemes to yield reduced-order models \citep{yan1999approximate}. The Riemannian trust-region method is used to solve the $H_2$ optimization problem over a Stiefel manifold in \citet{sato2015riemannian}. The tangential interpolation methods are based on moment matching and use efficient Krylov based algorithms, which include Iterative Rational Krylov Algorithm (IRKA) \citep{gugercin2008h_2} and Two-Sided Tangential Interpolation (TSIA) \citep{xu2011optimal}. Both methods are iterative and obtain reduced models satisfying first order $H_2$ optimality conditions. The methods are computationally efficient and applicable to large scale systems. However, there is no guarantee of the preservation of stability and convergence of the iterations. Modified versions of IRKA are proposed in \citet{beattie2007krylov}, \citet{beattie2009trust} to overcome such drawbacks. An iteration free pseudo optimal Rational Krylov (PORK) algorithm is proposed in \citet{wolf2014h}. It satisfies a subset of the first order $H_2$ optimality conditions. The Cumulative Reduction Scheme (CURE) proposed in \citet{panzer2014model} is used to construct reduced order models adaptively. 

The model reduction methods discussed till now yield good reduced-order models for infinite time-horizon. In specific settings, one may have access to simulation data over a finite time horizon, or one might be interested in approximating the output trajectory of the original system only over a limited time interval. In such cases, the model order reduction problem is restricted over a finite time interval. Accuracy outside the time interval is not essential.

Finite-time model reduction schemes include Proper Orthogonal Decomposition (POD) \citep{holmes2012turbulence} and time-limited balanced truncation (TL-BT) \citep{gawronski1990model}. Several variants of time-limited balanced truncation are proposed in \citet{gugercin2003time}, \citet{kumar2017generalized}, \citet{shaker2013generalized}, \citet{tahavori2013model}, \citet{jazlan2015cross}. A modified TL-BT algorithm is proposed in \citet{gugercin2003time} which preserves stability and an apriori error bound exits. But it has low efficiency and is computationally heavy compared to TL-BT. An output error bound for TL-BT is proposed in \citet{redmann2018output}. An $L_{\tau}^2$ error bound for TL-BT is suggested in \citet{redmann2020lt2} which converges to the well known $H_{\infty}$ error bound for BT as $\tau \to \infty$. Time-limited balanced truncation is studied in \citet{kurschner2018balanced} for large-scale continuous LTI systems and in \citet{duff2021numerical} for large-scale discrete LTI systems.

In \citet{goyal2019time}, Lyapunov based $H_2(\tau)$ optimality conditions are derived and an iterative scheme to obtain $H_2(\tau)$ optimal reduced-order models is proposed. Using a diagnostic measure involving the errors in the Lyapunov based $H_2(\tau)$ optimality conditions, the nearness to optimality is quantified. The algorithm proposed can be seen as a time-limited version of TSIA \citep{xu2011optimal}. In our paper, we shall refer to their algorithm as TL-TSIA. Interpolation based first-order necessary conditions for time-limited $H_2$-optimal model order reduction are derived in \citet{sinani2019h2}. The authors propose a descent-based iterative algorithm (FHIRKA) for satisfying the optimality conditions, valid only for SISO systems.  \citet{zulfiqar2020time} proposes an iteration free time-limited pseudo-optimal rational Krylov algorithm (TL-PORK) which satisfies a subset of the $H_2(\tau)$ optimality conditions and an adaptive version of TL-PORK known as TL-CURE which yields high-fidelity reduced-order models within a specific time interval. 

Interpolatory techniques for model order reduction yield high fidelity reduced-order models. Further, they are numerically efficient. Though interpolation-based $H_2(\tau)$-optimality conditions have been derived in \cite{sinani2019h2}, an interpolatory approach analogous to the IRKA for $H_2(\tau)$-optimal model order reduction is not available. We first derive the interpolation-based $H_2(\tau)$ optimality conditions along the lines of the work in \cite{breiten2015near}. 
Then we propose an interpolation based $H_2(\tau)$ optimal model reduction algorithm called Limited Time Iterative Rational Krylov Algorithm (LT-IRKA). The algorithm is iterative and uses a pair of modified finite time rational Krylov subspaces to construct the reduced-order model at every iteration. It minimizes the $H_2(\tau)$ error and yields a near-optimal reduced-order system. The tangential and bi-tangential errors involved in the finite time rational Krylov based model reduction are quantified. They are used to measure the closeness of the reduced system to the $H_2(\tau)$ optimal reduced system. The quantification can be used to explain the scenarios where LT-IRKA will yield good reduced models. We show that LT-IRKA and TL-TSIA produce similar left and right projection spaces and hence produce similar near $H_2(\tau)$ optimal reduced-order models. This is verified from numerical experiments where TL-TSIA and LT-IRKA result in reduced models with similar $H_2(\tau)$ errors. We also compare the performance of LT-IRKA with TL-BT, IRKA and TL-PORK using standard examples for various time intervals. 

We discuss some mathematical preliminaries in Section \ref{Prelims}.
In Section \ref{Section 2}, we derive time-limited $H_2(\tau)$ optimality conditions using a strategy similar to \cite{breiten2015near}. In Section \ref{Section 3} we discuss a variation of the rational Krylov based model reduction, modified to result in the tangential interpolation of the time-limited transfer function. The expressions for interpolation errors are derived, and conditions needed for good interpolation are analysed. Based on the $H_2(\tau)$ optimality conditions and the projection-based time-limited Krylov based model reduction, an iterative model reduction strategy along the lines of IRKA is proposed in Section \ref{Section 4}, and its relation to TL-TSIA is considered. Section \ref{Section 5} demonstrates the working of the proposed algorithm using three benchmark LTI systems and comparisons with existing strategies in the literature. We conclude the paper in Section \ref{Section 6} with a summary.  

\section{Preliminaries}\label{Prelims}

Consider a continuous linear time-invariant(LTI) system, $\Sigma$ with state-space representation :
\begin{subequations}\label{eq1}
\begin{equation} 
\dot{x}(t) = Ax(t)+Bu(t), \, x(0)= 0,\\
\end{equation}
\begin{equation} 
y(t) = Cx(t), \, t \geq 0   
\end{equation}
\end{subequations}
 where $A \in \mathbb{R}^{n \times n}$, $B \in \mathbb{R}^{n \times m}$ and $C \in \mathbb{R}^{p \times n}$.  The state dimension $n$ is a relatively large. Let $g(t) \in \mathbb{R}^{(p \times m)}$ be the impulse response and $G(s)$ be the transfer function of $\Sigma$. Consider $A$ to be Hurwitz, implying that the system (\ref{eq1}) is asymptotically stable.

The focus in MOR is to a obtain reduced order system, $\hat{\Sigma}$ as follows:
\begin{subequations}\label{eq2} 
\begin{equation} 
\dot{\hat{x}}(t) = \hat{A}\hat{x}(t)+\hat{B}\hat{u}(t), \hat{x}(0)= 0, \\
\end{equation}
\begin{equation} 
\hat{y}(t) = \hat{C}\hat{x}(t), \, t \geq 0   
\end{equation}
\end{subequations}
where $\hat{A} \in \mathbb{R}^{r \times r}$, $\hat{B} \in \mathbb{R}^{r \times m}$  and $\hat{C} \in \mathbb{R}^{p \times r}$ with $r \ll n$. Let $\hat{G}(s)$ be transfer function and $\hat{g}(t)$ be the impulse response of the reduced system.

For the model order reduction problem over the unrestricted time interval $[0,\infty]$, it is desirable that $y_r(t) \approx y(t)$ for all permissible inputs $u(t)$. This is achieved by ensuring that $(g-\hat{g})$ is small with respect to the $H_2$ system norm defined as,

\begin{equation}
\left\Vert g-\hat{g} \right\Vert_{H_2} = \left(\int_{0}^{\infty}\left\Vert g(t)-\hat{g}(t) \right\Vert_{F}^2 dt \right)^{\frac{1}{2}} 
\end{equation}
where the Frobenius norm $\left\Vert f \right \Vert_{F}$ for $f(t)$ over the time-interval $[0,\infty]$ is defined as 

\begin{equation}
\left\Vert f \right \Vert_{F}^{2} = \int_{0}^{\infty}\text{Tr}\left(f^{*}(t)f(t)\right) dt    
\end{equation}
The optimal $H_2$ model reduction problem involves finding a reduced $r$th order dynamical system $\hat{\Sigma}$ which is the best system of order $r$ approximating $\Sigma$ with respect to the $H_2$ norm and is defined as follows,
\begin{equation}\label{H2 optimal model reduction}
    \left\Vert \Sigma -\hat{\Sigma} \right\Vert_{H_2} = \operatorname*{min}_{\text{dim}(\Sigma_r)=r} \left\Vert \Sigma-\Sigma_r \right\Vert_{H_2}
\end{equation}
The $H_2$-optimization problem over the unrestricted time interval $[0,\infty]$ is non-convex and finding global minimizers is a difficult task. The usual approach is to propose algorithms that yield reduced order models satisfying local(first-order) necessary conditions for $H_2$ optimality which are mainly of two types - Lyapunov based \citep{wilson1970optimum} and Interpolation based \citep{meier1967approximation}. The reduced order models obtained are efficient in spite of the use of local minimizers. We state the Interpolation based $H_2$ optimality conditions and dicuss about IRKA. Refer to Theorem 3.4 in \citet{gugercin2008h_2} and Theorem 1 in \citet{antoulas2010interpolatory} for a detailed discussion.

\begin{theorem}
Consider the full-order system (\ref{eq1}) with impulse response $g(t) \in 
\mathbb{R}^{p \times m}, t \in [0,\infty)$. Let $\hat{g}(t) \in \mathbb{R}^{p \times m}, t \in [0,\infty]$ be the impulse response of the reduced order system represented by (\ref{eq2}) which solves the $H_2$ optimization problem (\ref{H2 optimal model reduction}). Assume that the reduced system is of order $r$ and is diagonalizable.
Let $\{\lambda_1,\cdots,\lambda_r\}$ be the $r$ simple poles of the transfer function of $\hat{G}(s)$. The impulse response matrix $\hat{g}(t)$ can be represented in the pole-residue form as,
\begin{align}\label{rt}
\hat{g}(t) &= \hat{C}e^{\hat{A}t}\hat{B} \\
     &= \sum_{k=1}^{r} \hat{c}_k{\hat{b}_k}^{*} e^{\lambda_k t}
\end{align}
with $\hat{A} \in \mathbb{R}^{r \times r}$, $\hat{B} \in \mathbb{R}^{r \times m} $, $\hat{C} \in \mathbb{R}^{p \times r}$, $\hat{c}_k \in \mathbb{C}^p $ and $\hat{b}_k \in \mathbb{C}^m$ for $k=1,\cdots,r$. Then for $k=1,2,\cdots,r$ we have
\begin{equation}\label{H2 Right Tangential}
G_{\tau}(-\lambda_k)\hat{b}_k  =  \hat{G}_{\tau}(-\lambda_k)\hat{b}_k
\end{equation}
\begin{equation}\label{H2 Left Tangential}
{\hat{c}_k}^{*}G_{\tau}(-\lambda_k)  =  {\hat{c}_k}^{*} \hat{G}_{\tau}(-\lambda_k)\\
\end{equation}
\begin{equation}\label{H2 Bi-Tangential}
{\hat{c}_k}^{*} G_{\tau}^{'}(-\lambda_k)\hat{b}_k  = { \hat{c}_k}^{*}  \hat{G}_{\tau}^{'}(-\lambda_k) \hat{b}_k
\end{equation}
where $G'(s)=\frac{d}{ds}G(s)$
\end{theorem}

\begin{theorem}

Consider a set of $r$ interpolation points $\{\sigma_1,\sigma_2,\hdots,\sigma_r\}$ and let $\{b_1,b_2,\hdots,b_r\}$ and $\{c_1,c_2,\hdots,c_r\}$ be the right and left tangential directions respectively. Define:
\begin{equation} \label{RtProjector}
    \text{Ran }(V) = \operatorname*{span}_{i=1,2,\hdots,r} \{ (\sigma_i I_n-A)^{-1}B b_i \}
\end{equation}
\begin{equation}\label{LtProjector}
    \text{Ran }(W) = \operatorname*{span}_{i=1,2,\hdots,r} \{ (\sigma_i I_n-A^{*})^{-1}C^{*}c_i \}
\end{equation}
Let $Z^{*} = (W^{*}V)^{-1}W^{*}$. The reduced system matrices are obtained as follows,
\begin{equation}\label{Reduced_Matrices}
\hat{A}=Z^{*}AV,\quad \hat{B}=Z^{*}B \quad \textrm{and}\quad \hat{C}=CV
\end{equation}
Assume that the interpolation points are not eigenvalues of $A$ or $\hat{A}$. Then the following holds,
\begin{equation}\label{H2 Right Interpolation}
    G(\sigma_i)b_i = \hat{G}(\sigma_i)b_i
\end{equation}
\begin{equation}\label{H2 Left Interpolation}
    {c_i}^{*}G(\sigma_i) = {c_i}^{*}\hat{G}(\sigma_i)
\end{equation}
\begin{equation}\label{H2 Bi-Interpolation}
    {c_i}^{*}G'(\sigma_i)b_i = {c_i}^{*}\hat{G}'(\sigma_i)b_i
\end{equation}
where $G'(s)=\frac{d}{ds}G(s)$.
\end{theorem}

Observe that the first-order $H_2$ optimality conditions are tangential interpolation conditions at mirror images of poles of the reduced order system. Since the interpolation points and tangential directions are not known apriori, constructing $V$ and $W$ is not easy. IRKA solves this issue by iteratively correcting the interpolation points and tangential directions. The major steps involved in IRKA are:
\begin{enumerate}
    \item For a set of $r$ random initial conditions and tangential directions, obtain the right and left projectors, $V$ and $W$ as in (\ref{RtProjector}) and (\ref{LtProjector}). Construct the reduced system matrices using (\ref{Reduced_Matrices}). Assuming $\hat{A}$ to be a diagonalizable matrix, let $R\Lambda R^{-1}$ be the eigenvalue decomposition of $\hat{A}$. Construct $\tilde{B}={\hat{B}}^{*}R^{-*}$ and $\tilde{C}=\hat{C}R$.
    
    \item  For the subsequent iterations, the negative of the eigenvalues of $\hat{A}$ are chosen as the interpolation points and the columns of $\tilde{B}$ and $\tilde{C}$ are chosen as the right and left tangential directions respectively in the next iteration.
    
    \item  The iterations are repeated till the norm of the difference between the interpolation points of two consecutive iterations becomes less than a certain tolerance value. Upon convergence, the reduced system satisfies the first-order $H_2$-optimality conditions (\ref{H2 Right Tangential}), (\ref{H2 Left Tangential}) and (\ref{H2 Bi-Tangential}). 
\end{enumerate}

IRKA is a Krylov-based model reduction method, and its implementation requires matrix-vector multiplications and some sparse linear solvers. Therefore it is computationally effective.

 \section{Interpolation based $H_2(\tau)$ optimality conditions}\label{Section 2}

\paragraph*{$H_2(\tau)$ inner product and norm: }
Let $H_2(\tau)$ be the set of $p\times m$ matrix valued functions $\{ g(t),\, t>0\}$ which are the impulse response matrices of finite dimensional dynamical systems. The functions are square integrable as a function of $t \in [0,\tau]$ in the sense that
\begin{equation}
\int_{0}^{\tau}\left\Vert g(t) \right\Vert_{F}^{2}dt < \infty
\end{equation}
We note that $H_2(\tau)$ is a Hilbert space. If $f(t)$ and $g(t)$ are impulse response matrices of dynamical systems such that $f,g \in H_2(\tau)$ then their inner product is defined as,
\begin{equation}\label{h_2 inner product}
\langle f,g \rangle_{H_2(\tau)} = \int_{0}^{\tau}\text{tr}(f(t)g^{*}(t)) dt
\end{equation}
where $g^{*}$ is the conjugate complex of $g$. The associated norm on $H_2(\tau)$ is given as
\begin{equation}\label{h_2 norm}
\left\Vert g \right\Vert_{H_2(\tau)} = \sqrt{\langle g,g \rangle_{H_2(\tau)}}
\end{equation}
As $\tau \to \infty$, the set $H_2(\tau)$ corresponds to  $H_2$ which consists of impulse response matrices of stable dynamical systems. Observe that $H_2 \subset H_2(\tau)$. For $f,g \in H_2(\tau)$, using the Cauchy-Schwarz inequality we can show that,
\begin{equation}
|\langle f,g \rangle_{H_2(\tau)}| \leq \left\Vert f \right\Vert_{H_2(\tau)} \left\Vert g \right\Vert_{H_2(\tau)} 
\end{equation}
The impulse response matrix associated with the system (\ref{eq1}) is $g $. Note that the system (\ref{eq1}) need not be stable for its impulse response $g$ to be in $H_2(\tau)$. Define the time restricted impulse response $g_{\tau}$ as 
\begin{equation}\label{gtau}
g_{\tau}(t)= \begin{cases}
          g(t), & t \in [0,\tau] \\
          0, & t \in (\tau,\infty)
         \end{cases}
\end{equation}

\paragraph*{Limited time transfer function ($G_{\tau}(s)$):}
For system (\ref{eq1}), $g(t)= Ce^{At}B,\, t \geq 0 $ is the impulse response and $G(s)= C (sI-A)^{-1} B$ is the transfer function. The Laplace transform of the time restricted impulse response matrix $g_{\tau}(t)$ (\ref{gtau}) is denoted by $G_{\tau}(s)$. It is derived as follows,
\begin{equation}\label{Time-limited Transfer Function_1}
\begin{aligned}
G_{\tau}(s) &= \int_{0}^{\infty} g_{\tau}(t)e^{-st}dt   \\
           &= -\int_{\tau}^{\infty}Ce^{At}Be^{-st}dt + G(s)\\   
           &= -Ce^{-s\tau}\left(\int_{0}^{\infty}e^{At}e^{-st}dt\right)e^{A\tau}B + G(s)\\
           &= -e^{-s\tau}C(sI_n-A)^{-1}e^{A\tau}B + G(s) 
\end{aligned}
\end{equation}
It can also be expressed as,
\begin{equation}\label{Time-limited Transfer Function_2}
G_{\tau}(s) =  -e^{-s\tau}Ce^{A\tau}(sI_n-A)^{-1}B + G(s) 
\end{equation}

\begin{proposition}\label{Proposition 1}
 Let $g_1(t)$ be the impulse response matrix of an LTI system with real state space realization and $G_1(s)$ be the corresponding transfer function. Let $g_{1,\tau}(t)$ be the restriction of $g_1(t)$ to the time-interval $[0,\tau]$ and $G_{1,\tau}(s)$ be the corresponding Laplace transform. Assume $g_2(t)=cb^{*}e^{\mu t}$ and $g_3(t)=cb^{*}te^{\mu t}$ for $t>0$. Then,

\begin{equation} \label{Prop11}
(1)\quad \langle g_1,g_2 \rangle_{H_2(\tau)} = c^{*}\overline{G_{1,\tau}(-\mu)}b
\end{equation}

\begin{equation} \label{Prop12}
(2)\quad  \left\Vert g_2 \right \Vert_{H_2(\tau)} = \frac{\left\Vert b \right\Vert \left\Vert c \right\Vert}{\sqrt{2|\text{Re}(\mu)|}}\sqrt{|1-e^{2\tau \text{Re}(\mu)}|} 
\end{equation}

\begin{equation}\label{Proop13}
(3)\quad  \langle g_1,g_3 \rangle_{H_2(\tau)} = -c^{*}\overline{G_{1,\tau}^{'}(-\mu)}b
\end{equation}
\end{proposition}
\begin{proof} 
\begin{enumerate}
\item By definition of the $H_2(\tau)$ inner product,
\begin{align*}
\langle g_1,g_2 \rangle_{H_2(\tau)} &= \int_{0}^{\tau} Tr(g_1(t)bc^{*}e^{\mu^{*}t})dt \\
                                    &= \int_{0}^{\tau} c^{*}g_1(t)be^{\mu^{*}t}dt \\
                                    & = c^{*}\left(\overline{\int_{0}^{\tau}g_1(t)e^{\mu t}dt}\right)b \\
                                    &= c^{*}\left(\overline{\int_{0}^{\infty}g_{1,\tau}(t)e^{\mu t}dt}\right)b \\
                                    &= c^{*}\overline{G_{1,\tau}(-\mu)}b
\end{align*}
The last step in the derivation follows by the definition of the Laplace Transform.

\item If we have $g_1=g_2$, then by definition of $H_2(\tau)$ norm  we have  
$\langle g_2,g_2 \rangle_{H_2(\tau)} = \left\Vert g_2 \right\Vert_{H_2(\tau)}^{2}$.
\begin{align*}
\left\Vert g_2 \right\Vert_{H_2(\tau)}^{2} &= \int_{0}^{\tau} Tr(cb^{*}e^{\mu t}bc^{*}e^{\mu^{*}t} )dt \\
                    &= \int_{0}^{\tau} c^{*}\left\Vert b \right\Vert^{2} c e^{(\mu+\mu^{*})t}dt  \\
                    &= \int_{0}^{\tau} \left\Vert b \right\Vert^{2} \left\Vert c \right\Vert^{2} e^{2Re(\mu)t}dt\\
                    &=  \frac{\left\Vert b \right\Vert^{2} \left\Vert c \right\Vert^{2}}{2|Re(\mu)|}|e^{2(Re(\mu))\tau}-1| 
\end{align*}
Thus, we get
\begin{equation}
\left\Vert g_2 \right\Vert_{H_2(\tau)} = \frac{\left\Vert b \right\Vert\left\Vert c \right\Vert}{\sqrt{2|Re(\mu)|}} \sqrt{|e^{2(Re(\mu))\tau}-1|} 
\end{equation}

\item The proof again follows by the definition of the $H_2(\tau)$ norm.
\begin{align*}
\langle g_1,g_3 \rangle_{H_2(\tau)} &= \int_{0}^{\tau}Tr(g_1(t)bc^{*}te^{\mu^{*}t})dt\\
  &= \int_{0}^{\tau}c^{*}g_1(t)te^{\mu^{*}t}b dt \\
  &= c^{*} \left(\overline{\int_{0}^{\tau}g_1(t)te^{\mu t}dt}\right) b \\
  &= c^{*} \left(\overline{\int_{0}^{\infty}g_{1,\tau}(t)te^{\mu t}dt}\right) b \\
  &= -c^{*} \overline{G_{1,\tau}^{'} (-\mu)} b
\end{align*}
The last step follows as a property of the Laplace transform. 
\end{enumerate}
\end{proof}

\paragraph*{$H_2(\tau)$ optimal model reduction:}
For a good approximation of the system (\ref{eq1}) by the reduced system (\ref{eq2}) over the time interval $[0, \tau]$, we require $\left\Vert y(t)-\hat{y}(t)\right\Vert_{H_2(\tau)}$ to be small. The norm $\left\Vert \cdot \right\Vert_{H_2(\tau)}$ is defined over the time interval $t \in [0, \tau]$. The following expression is derived in \citet{goyal2019time}.
\begin{align}
\max_{t \in [0,\tau]} \left\Vert y(t)-\hat{y}(t) \right\Vert_2 &\leq \left\Vert g-\hat{g} \right\Vert_{H_2(\tau)}\left\Vert u \right\Vert_{L_2^{\tau}} 
\end{align}
where $ \left\Vert g-\hat{g} \right\Vert_{H_2(\tau)} = \left(\bigintss_{0}^{\tau}\left\Vert g(t)-\hat{g}(t) \right\Vert_{F}^{2}dt\right)^{\frac{1}{2}}$. Thus, minimizing this norm ensures that the reduced order model output $\hat{y}(t)$ is close to the original output $y(t)$ over the time interval $[0,\tau]$. 

The $H_2(\tau)$-optimal model reduction problem can be written as,
\begin{equation}\label{H2tau optimization}
\hat{\Sigma} = \operatorname*{arg\,min}_{\text{dim }(\Sigma_r) = r} \left\Vert \Sigma-\Sigma_r \right\Vert_{H_2(\tau)}
\end{equation}
The $H_2(\tau)$ optimization problem (\ref{H2tau optimization}) is non-convex and finding global minimizers will be difficult. In the theorem below, we derive the interpolation based first-order necessary conditions for $H_2(\tau)$ optimality of the reduced order system, $\hat{\Sigma}$. 

\begin{theorem}\label{Theorem 1}
Consider the full-order system (\ref{eq1}) with impulse response $g(t) \in 
\mathbb{R}^{p \times m}, t \in [0,\infty)$. Let $\hat{g}(t) \in \mathbb{R}^{p \times m}, t \in [0,\infty]$ be the impulse response of the reduced order system represented by (\ref{eq2}) which also solves the $H_2(\tau)$ optimization problem (\ref{H2tau optimization}). Assume that the reduced system is of order $r$ and is diagonalizable.
Let $\{\lambda_1,\cdots,\lambda_r\}$ be the $r$ simple poles of the transfer function of $\hat{G}(s)$. The impulse response matrix $\hat{g}(t)$ can be represented in the pole-residue form as,
\begin{align}\label{rt}
\hat{g}(t) &= \hat{C}e^{\hat{A}t}\hat{B} \\
     &= \sum_{k=1}^{r} \hat{c}_k{\hat{b}_k}^{*} e^{\lambda_k t}
\end{align}
with $\hat{A} \in \mathbb{R}^{r \times r}$, $\hat{B} \in \mathbb{R}^{r \times m} $, $\hat{C} \in \mathbb{R}^{p \times r}$, $\hat{c}_k \in \mathbb{C}^p $ and $\hat{b}_k \in \mathbb{C}^m$ for $k=1,\cdots,r$. Let $G_{\tau}(s)$ and $\hat{G}_{\tau}(s)$ be the Laplace transforms of $g_{\tau}(t)$ and $\hat{g}_{\tau}(t)$ respectively. Then for $k=1,2,\cdots,r$,
\begin{equation}\label{Time Limited Right Tangnetial Interpolation}
G_{\tau}(-\lambda_k)\hat{b}_k  =  \hat{G}_{\tau}(-\lambda_k)\hat{b}_k
\end{equation}
\begin{equation}\label{Time Limited Left Tangential Interpolation}
{\hat{c}_k}^{*}G_{\tau}(-\lambda_k)  =  {\hat{c}_k}^{*} \hat{G}_{\tau}(-\lambda_k)\\
\end{equation}
\begin{equation}\label{Time Limited Bi-Tangential Interpolation}
{\hat{c}_k}^{*} G_{\tau}^{'}(-\lambda_k)\hat{b}_k  = { \hat{c}_k}^{*}  \hat{G}_{\tau}^{'}(-\lambda_k) \hat{b}_k
\end{equation}
\end{theorem}

\begin{proof} 
Let $a \in \mathbb{C}^p$ be an arbitrary vector with $\left\Vert a \right\Vert=1$ and an index $k$ where $1 \leq k \leq r$. Assume that,
\begin{equation}
\langle g-\hat{g}, a{\hat{b}_k}^{*}e^{\lambda_k t} \rangle_{H_2(\tau)} = \alpha (\neq 0)
\end{equation}
Let $\arg(\alpha)=\theta_0$ and for some arbitrary $\epsilon > 0$ we define the following perturbation to $\hat{g}$,
\begin{equation}\label{rtep}
\hat{g}^{\epsilon}(t)=(\hat{c}_k+\epsilon e^{i\theta_0}a){\hat{b}_k}^{*}e^{\lambda_k t}+\sum_{i \neq k}\hat{c}_i {\hat{b}_i}^{*}e^{\lambda_i t}
\end{equation}
Using (\ref{rt}), (\ref{rtep}) and (\ref{Prop12}), 
\begin{align*}
\left\Vert \hat{g}-\hat{g}^{\epsilon} \right\Vert_{H_2(\tau)} &= \left\Vert -\epsilon^{-i\theta_0}a{\hat{b}_k}^{*}e^{\lambda_k t} \right\Vert_{H_2(\tau)} \\
       &= \epsilon \frac{\left\Vert a \right\Vert \left\Vert {{\hat{b}}_k}^{*} \right\Vert }{\sqrt{2|Re(\mu)|}}\sqrt{|1-e^{2\tau Re(\mu)}|} \\
       &= \epsilon \frac{ \left\Vert {{\hat{b}}_k}^{*} \right\Vert }{\sqrt{2|Re(\mu)|}}\sqrt{|1-e^{2\tau Re(\mu)}|} 
\end{align*}
Thus,
\begin{equation*}
\left\Vert \hat{g}-\hat{g}^{\epsilon} \right\Vert = \mathcal{O}(\epsilon) \qquad as \qquad \epsilon \to 0
\end{equation*}
As $\hat{g}$ solves the $H_2(\tau)$ optimization problem (\ref{H2tau optimization}) the following relation holds true,
\begin{align*}
\left\Vert g-\hat{g} \right\Vert_{H_2(\tau)}^2  &\leq \left\Vert g-\hat{g}^{\epsilon} \right\Vert_{H_2(\tau)}^2 \\
                                        &\leq \left\Vert (g-\hat{g})+(\hat{g}-\hat{g}^{\epsilon}) \right\Vert_{H_2(\tau)}^2\\
                                        &\leq  \left\Vert g-\hat{g} \right\Vert_{H_2(\tau)}^2 + 2\textrm{Re}\langle g-\hat{g},\hat{g}-\hat{g}^{\epsilon} \rangle_{H_2(\tau)}+ \\
&  \qquad \left\Vert \hat{g}-\hat{g}^{\epsilon} \right\Vert_{H_2(\tau)}^2
\end{align*}
Hence,
\begin{equation}
0 \leq 2\textrm{Re}\langle g-\hat{g},\hat{g}-\hat{g}^{\epsilon} \rangle_{H_2(\tau)}+\left\Vert \hat{g}-\hat{g}^{\epsilon} \right\Vert_{H_2(\tau)}^2
\end{equation}
Note that,
\begin{align*}
 & 2\textrm{Re}\langle g-\hat{g},\hat{g}-\hat{g}^{\epsilon} \rangle_{H_2(\tau)} \\ &= 2\textrm{Re}(\overline{-\epsilon e^{i\theta_0}}\langle g-\hat{g},a{b_k}^{*}e^{\lambda_k t} \rangle)_{H_2(\tau)} \\ &= 2\textrm{Re}(-\epsilon e^{i\theta_0}\alpha) \qquad (\textrm{By definition}) \\
&= -2\epsilon|\alpha|
\end{align*}
The above discussion implies that $0 \leq -\epsilon|\alpha|+\mathcal{O}(\epsilon^2)$. This is not true for arbitrary $\alpha$ and is possible only if $\alpha=0$. This will lead to the following result,
\begin{align*}
0 &= \langle g-\hat{g},a{\hat{b}_k}^{*}e^{\lambda_k t} \rangle_{H_2(\tau)}\\
  &= a^{*}\overline{(G_\tau-\hat{G}_{\tau})}(-\lambda_k)\hat{b}_k \quad (\textrm{ Using (\ref{Prop11})})\\
\end{align*}
As $g$ was chosen arbitrarily, we have
\begin{align*}
&{(G_\tau-\hat{G}_{\tau})}(-\lambda_k)\hat{b}_k =0 \\
& G_{\tau}(-\lambda_k)\hat{b}_k = \hat{G}_{\tau}(-\lambda_k)\hat{b}_k 
\end{align*}
This proves (\ref{Time Limited Right Tangnetial Interpolation}). We prove (\ref{Time Limited Left Tangential Interpolation}) in a similar way by repeating the above analysis with $\hat{c}_k d^{*}e^{\lambda_k t}$ instead of $a{\hat{b}_k}^{*}e^{\lambda_k t}$ for some arbitrary $d \in \mathbb{C}^m$ and $\left\Vert d \right\Vert=1$.

In order to prove (\ref{Time Limited Bi-Tangential Interpolation}), we assume that $\langle g-\hat{g},\hat{c}_k{\hat{b}_k}^{*}t e^{\lambda_k t} \rangle_{H_2(\tau)}=\beta(\neq 0)$. Let $\theta_1=\arg(\beta)$. For sufficiently small $\epsilon>0$,
\begin{align}
\hat{g}^{\epsilon}(t)=\hat{c}_k{\hat{b}_k}^{*}e^{(\lambda_k+\epsilon e^{-i\theta_1})t}+\sum_{i \neq k}\hat{c}_i{\hat{b}_i}^{*}e^{\hat{\lambda}_i t}
\end{align} 
For the above expression of $\hat{g}^{\epsilon}(t)$, the following holds
\begin{align*}
&\quad \, \left\Vert \hat{g}-\hat{g}^{\epsilon} \right\Vert_{H_2(\tau)}  \\
&= \left\Vert 
\hat{c}_k{\hat{b}_k}^{*}e^{\lambda_k t}(1-e^{\epsilon e^{-i\theta_1}t }) \right\Vert \\
&= \left\Vert \hat{c}_k{\hat{b}_k}^{*}e^{\lambda_k t} ( -\epsilon e^{-i\theta_1}t    +{\epsilon}^2 e^{-2i\theta_1}t^2+\cdots)\right\Vert  \\
&= \mathcal{O}(\epsilon)
\end{align*}
Following the same steps as above we get,
\begin{equation}
0 \leq 2\textrm{Re}\langle g-\hat{g},\hat{g}-\hat{g}^{\epsilon} \rangle_{H_2(\tau)}+\left\Vert \hat{g}-\hat{g}^{\epsilon} \right\Vert_{H_2(\tau)}^2
\end{equation}
and
\begin{align*}
& 2\textrm{Re}\langle g-\hat{g},\hat{g}-\hat{g}^{\epsilon} \rangle_{H_2(\tau)}\\
&=2\textrm{Re}\langle g-\hat{g},\hat{c}_k{\hat{b}_k}^{*}e^{\lambda_k t} ( -\epsilon e^{-i\theta_1}t ) \rangle_{H_2(\tau)}  \\
&= 2\textrm{Re}(-\epsilon e^{i\theta_1}\langle g-\hat{g},\hat{c}_k{\hat{b}_k}^{*}te^{\lambda_k t} \rangle_{H_2(\tau)})\\
&= 2\textrm{Re}(-\epsilon e^{i\theta_1}\langle g-\hat{g},\hat{c}_k{\hat{b}_k}^{*}te^{\lambda_k t} \rangle_{H_2(\tau)})\\
&= -2\epsilon |\beta|
\end{align*} 
The above discussion implies that $0 \leq -2\epsilon|\beta|+\mathcal{O}(\epsilon^2)$. This is not true for arbitrary $\beta$ and is possible only if $\beta=0$. 

Now, (\ref{Proop13}) leads to the following result,
\begin{align*}
0 &= \langle g-\hat{g},\hat{c}_k{\hat{b}_k}^{*}t e^{\lambda_k t} \rangle_{H_2(\tau)}\\
  &= -{\hat{c}_k}^{*}(G_\tau-\hat{G}_\tau)^{'}(-\lambda_k)\hat{b}_k
\end{align*}
which we get (\ref{Time Limited Bi-Tangential Interpolation}).
\end{proof}

\begin{remark}
Note that our method is different from the one used in \cite{sinani2019h2} to obtain the $H_2(\tau)$ optimality conditions and is similar to the procedure used to obtain frequency weighted $H_2$-optimality conditions in \cite{breiten2015near}.
\end{remark}

\begin{remark}
Theorem 1 says that the reduced system $\hat{\Sigma}$ is a $H_2(\tau)$ optimal approximation of the original system $\Sigma$ with respect to the $H_2(\tau)$ norm among all reduced systems having the same reduced system poles $\{ \hat{\lambda}_i \}_{i=1}^{r}$. The set of all systems having the same poles $\{ \hat{\lambda}_i \}_{i=1}^{r}$ is a subspace of $H_2(\tau)$ and the projection theorem gives necessary and sufficient condition for $\hat{g}$ to be a minimizer out of a subspace of candidate minimizers.
\end{remark}

\section{Rational Interpolation over a restricted time-interval}\label{Section 3}

Rational Krylov based model order reduction over a restricted time interval $[0,\tau]$ involves the computation of a pair of modified Rational Krylov spaces used to construct left and right projectors. A reduced order model is obtained via the Petrov-Galerkin projection method. We construct the projection subspaces for time-limited approximate rational interpolation as follows, 

\begin{equation}\label{RightProjectionSpace}
\mathcal{V}_r = \spann\limits_{i=1,\hdots,r}\{(\sigma_i I_n-A)^{-1}(I_n-e^{-\sigma_i\tau}e^{A\tau})Bb_i  \}
\end{equation}
and
\begin{equation}\label{LeftProjectionSpace}
\mathcal{W}_r = \spann\limits_{i=1,\hdots,r}\{(\sigma_i I_n-A^{*})^{-1}(I_n-e^{-\sigma_i\tau}e^{A^{*}\tau})C^{*}c_i  \}
\end{equation}
where $I_n$ is the $n\times n$ identity matrix.

We define matrices $V_r,W_r \in \mathbb{C}^{n \times r}$ such that
\begin{equation}\label{RightProjector}
\mathcal{V}_r \subset \textrm{Ran}(V_r)
\end{equation}
\begin{equation}\label{LeftProjector}
\mathcal{W}_r \subset \textrm{Ran}(W_r)
\end{equation}
The matrices $V_r$ and $W_r$ are the left and right projectors and a reduced order model is obtained via Petrov-Galerkin projection. The modified Krylov projection doesn't result in exact rational interpolation over a finite time interval. The deviation is quantified in Theorem \ref{Theorem 2}.

\begin{theorem} \label{Theorem 2}
 Let ${Z_r}^{*}=({W_r}^{*}V_r)^{-1}{W_r}^{*}$ and $\hat{A} = {Z_r}^{*}A{V_r}$, $\hat{B}={Z_r}^{*}B$, $\hat{C} = CV_r$ where $V_r$ and $W_r$ are obtained using (\ref{RightProjectionSpace}), (\ref{LeftProjectionSpace}), (\ref{RightProjector}) and (\ref{LeftProjector}). We construct a projection matrix $\Pi =V_r{Z_r}^{*}$.  Here, we assume that $D=0$. Let us choose any interpolation point $\sigma \in \{\sigma_1,\hdots,\sigma_r \}$ and tangent directions: $b$ and $c$. 
Assume that $\sigma \in \mathds{C}$ is not an eigenvalue of  A or $\hat{A}$. Then
\begin{equation}\label{Right Tangential Error}
\begin{aligned}
& G_{\tau}(\sigma)b-\hat{G}_{\tau}(\sigma)b \\
& = e^{-\sigma \tau}C{V_r}(\sigma I_r-\hat{A})^{-1}{Z_r}^{*}(e^{A\Pi\tau}-e^{A\tau})Bb
\end{aligned}
\end{equation}
\begin{equation}\label{Left Tangential Error}
\begin{aligned}
& c^{*}G_{\tau}(\sigma)-c^{*}\hat{G}_{\tau}(\sigma) \\
& = e^{-\sigma \tau}c^{*}C(e^{\Pi A \tau}-e^{A \tau})V_r(\sigma I_r-\hat{A}){Z_r}^{*}B
\end{aligned}
\end{equation}
\begin{equation}\label{Bi-Tangential ErrorP}
c^{*}{G_{\tau}}'(\sigma)b - c^{*}{\hat{G}_{\tau}}'(\sigma)b = R_{P_1}(\sigma) + R_{P_2}(\sigma) 
\end{equation}
\begin{equation}\label{Bi-Tangential ErrorQ}
c^{*}{G_{\tau}}'(\sigma)b - c^{*}{\hat{G}_{\tau}}'(\sigma)b = R_{Q_1}(\sigma)+R_{Q_2}(\sigma)
\end{equation}

Here, $R_{P_1}(\sigma), R_{P_2}(\sigma)$, $R_{Q_1}(\sigma),R_{Q_2}(\sigma)$ are,
\begin{align*}
R_{P_1}(\sigma) &= -e^{-\sigma \tau}c^{*}C {V_r}(\sigma I_r-\hat{A})^{-2}((\sigma I_r-\hat{A})\tau \\ 
& +I_r){Z_r}^{*}(e^{A\Pi\tau}-e^{A\tau})Bb \\
R_{P_2}(\sigma) &= e^{-\sigma \tau}c^{*}Ce^{A\tau}(I_n-P(\sigma))(\sigma I_n-A)^{-2}\\
& ((\tau(\sigma I_n-A)+I_n) e^{-\tau(\sigma I_n-A)}-I_n)Bb\\
R_{Q_1}(\sigma) &= -e^{-\sigma \tau}c^{*}C(e^{\Pi A \tau}-e^{A \tau})V_r(\sigma I_r-\hat{A})^{-2}\\ 
& (I_r+\tau(\sigma I_r-\hat{A})){Z_r}^{*}Bb\\
R_{Q_2}(\sigma) &= e^{-\sigma \tau}c^{*}C((\tau(\sigma I_n-A)+I_n) e^{-\tau(\sigma I_n-A)}-I_n)\\
&(\sigma I_n-A)^{-2}(I_n-Q(\sigma))e^{A\tau}Bb
\end{align*}
$P(s)$ and $Q(s)$ are projector matrices defined as,
\begin{align*}
P(s) &= V(s I_r-\hat{A})^{-1}Z^{*}(s I_n-A) \\
Q(s) &= (s I_n-A)V(s I_r-\hat{A})^{-1}Z^{*} 
\end{align*}
\end{theorem}

\begin{proof}
Based on the assumption that $\sigma$ is not not eigenvalue of $A$ and $\hat{A}$, $P(s)$ and $Q(s)$ are matrix-valued analytic functions for all $s$ in a neighbourhood of $s=\sigma$. Note that, $P^2(s)=P(s)$ and $Q^2(s)=Q(s)$. Also, we may verify that $\mathcal{V}=\mathrm{Ran  }P(s)=\mathrm{Ker}(I_n-P(s))$ and ${\mathcal{W}}^{\perp}=\mathrm{Ker}Q(s)=\mathrm{Ran}(I_n-Q(s))$. 

The error in the right tangential interpolation condition is derived as follows, 
\begin{equation}\label{Right Interpolation condition1}
\begin{aligned}
& G_{\tau}(s)b-\hat{G}_{\tau}(s)b \\                  &=C(sI_n-A)^{-1}(I_n-e^{-s\tau}e^{A\tau})Bb-\\
&\hat{C}(sI_r-\hat{A})^{-1}(I_r-e^{-s\tau}e^{\hat{A}\tau})\hat{B}b
\end{aligned}
\end{equation}
Substituting $\hat{A}={Z_r}^{*}A{V_r}$, $\hat{B}={Z_r}^{*}B$, $\hat{C}=C{V_r}$ and using the identity $e^{\hat{A}\tau}\hat{B}={Z_r}^{*}e^{A\Pi \tau}$(here $\Pi=V_r {Z_r}^{*}$), the second term in the RHS of equation (\ref{Right Interpolation condition1}) can be rewritten in the following way,
\begin{align*}
  & \hat{C}(sI_r-\hat{A})^{-1}(I_r-e^{-s\tau}e^{\hat{A}\tau})\hat{B}b\\
  &= \hat{C}(sI_r-\hat{A})^{-1}(\hat{B}-e^{-s\tau}e^{\hat{A}\tau}\hat{B})b \\
  &= CV_r (sI_r-\hat{A})^{-1}({Z_r}^{*}B-{Z_r}^{*}e^{-s\tau}e^{A\Pi \tau}B)b \\
  &= CV_r (sI_r-\hat{A})^{-1}{Z_r}^{*}(I_n-e^{-s\tau}e^{A\Pi \tau})Bb \\
  &=  CV_r (sI_r-\hat{A})^{-1}{Z_r}^{*}(I_n-e^{-s\tau}e^{A \tau})Bb \\
  &+ e^{-s\tau}CV_r (sI_r-\hat{A})^{-1}{Z_r}^{*}(e^{A\tau}-e^{A\Pi\tau})Bb \\
  &= CV_r (sI_r-\hat{A})^{-1}{Z_r}^{*}(I_n-e^{-s\tau}e^{A \tau})Bb \\
  &+ e^{-s\tau}CP(s)(sI_n-A)^{-1}(e^{A\tau}-e^{A\Pi\tau})Bb
\end{align*}  
Substituting the above expression in (\ref{Right Interpolation condition1}) we get,
\begin{equation}\label{Right Interpolation Condition3}
\begin{aligned}
& G_{\tau}(s)b-\hat{G}_{\tau}(s)b \\
& = C[I_n-P(s)](sI-A)^{-1}(I_n-e^{-s\tau}e^{A\tau})Bb\\
& +e^{-s\tau}CV_r(s I_r-\hat{A})^{-1}{Z_r}^{*}(e^{A\Pi \tau}-e^{A\tau})Bb 
\end{aligned}
\end{equation}
Evaluating the expression (\ref{Right Interpolation Condition3}) at $s=\sigma$ we get (\ref{Right Tangential Error}).

Replacing $\hat{A}$, $\hat{B}$, $\hat{C}$ similar to the case of right tangential interpolation error and using the identity $\hat{C}e^{\hat{A}\tau}=Ce^{\Pi A \tau}V$  we derive the left tangential interpolation error,

\begin{equation}\label{Left Interpolation condition1}
\begin{aligned}
&c^{*}G_{\tau}(s)-c^{*}\hat{G}_{\tau}(s)\\
&=  c^{*}C(I_n-e^{-s\tau}e^{A\tau})(s I_n-A)^{-1}B-\\
&c^{*}\hat{C}(I_r-e^{-s\tau}e^{\hat{A}\tau})(s I_r-\hat{A})^{-1}\hat{B}
\end{aligned}
\end{equation}

The second term in the RHS of the above equation can be expressed as,
\begin{align*} 
& c^{*}\hat{C}(I_r-e^{-s\tau}e^{\hat{A}\tau})(sI_r-\hat{A})^{-1}\hat{B} \\
&= c^{*}(\hat{C}-e^{-s\tau}\hat{C}e^{\hat{A}\tau})(sI_r-\hat{A})^{-1}\hat{B} \\
&= c^{*}(CV_r-e^{-s\tau}Ce^{\Pi A\tau}V_r)(sI_r-\hat{A})^{-1}{Z_r}^{*}B \\
&= c^{*}C(I_n-e^{-s\tau}e^{\Pi A\tau})V_r(sI_r-\hat{A})^{-1}{Z_r}^{*}B \\
&= c^{*}C(I_n-e^{-s\tau}e^{A\tau})V_r(sI_r-\hat{A})^{-1}{Z_r}^{*}B+ \\
& e^{-s\tau}c^{*}C(e^{A\tau}-e^{\Pi A\tau})V_r(sI_r-\hat{A})^{-1}{Z_r}^{*}B \\
&=  c^{*}C(I_n-e^{-s\tau}e^{A\tau})(sI_n-A)^{-1}Q(s)B +\\
&  e^{-s\tau}c^{*}C(e^{A\tau}-e^{\Pi A\tau})V_r(sI_r-\hat{A})^{-1}{Z_r}^{*}B
\end{align*}
Substituting the above expression in (\ref{Left Interpolation condition1}) and representing $I_n-e^{-s\tau}e^{A\tau}$ as $E(s)$ we get,
\begin{equation}\label{Left Interpolation condition3}
\begin{aligned}
& c^{*}G_{\tau}(s)-c^{*}\hat{G}_{\tau}(s) = \\ 
& ((I_n-e^{-s\tau}e^{A^{*}\tau})(sI_n-A^{*})^{-1}C^{*}c)^{*}[I-Q(s)]B\\  &+e^{-s\tau}c^{*}C(e^{\Pi A\tau}-e^{A\tau})V_r(s I_r-\hat{A})^{-1}{Z_r}^{*}B
\end{aligned}
\end{equation}
We have to evaluate (\ref{Left Interpolation condition3}) at $s=\sigma$ to obtain (\ref{Left Tangential Error}).

Left multiplying by $c^{*}$ and after certain rearrangements of equation (\ref{Right Interpolation Condition3}) we get,
\begin{equation}\label{Bi tangential Interpolation condition1}
\begin{aligned}
& c^{*}G_{\tau}(s)b-c^{*}\hat{G}_{\tau}(s)b  \\
&= c^{*}CE(s)(sI_n-A)^{-1}(I_n-Q(s))\\
& (sI_n-A)(I-P(s))(sI_n-A)^{-1}E(s)Bb \\
& + c^{*}Ce^{-s\tau}e^{A\tau}(I_n-P(s))(sI_n-A)^{-1}E(s)Bb \\
& e^{-s\tau}c^{*}CV_r(sI_r-\hat{A})^{-1}{Z_r}^{*}(e^{A\Pi\tau}B-e^{A\tau}B)b  
\end{aligned}
\end{equation}

Differentiating (\ref{Bi tangential Interpolation condition1}) with respect to $s$ and evaluating them at $s=\sigma$ we get the two terms given by 
(\ref{Bi-Tangential ErrorP}).  

Right multiplying by $b$ and after certain rearrangements of equation (\ref{Left Interpolation condition3}) we get,
\begin{equation}\label{Bi tangential Interpolation condition2}
\begin{aligned}
& c^{*}G_{\tau}(s)b-c^{*}\hat{G}_{\tau}(s)b \\ 
&= c^{*}CE(s)(sI_n-A)^{-1}(I_n-Q(s))\\
&(sI_n-A)(I-P(s))(sI_n-A)^{-1}E(s)Bb \\
&+ c^{*}CE(s)(sI_n-A)^{-1}(I_n-Q(s))e^{-s\tau}e^{A\tau}Bb \\
&+  e^{-s\tau}c^{*}C(e^{\Pi A\tau}-e^{A\tau})V_r(sI_r-\hat{A})^{-1}{Z_r}^{*}Bb   
\end{aligned}
\end{equation}

Differentiating (\ref{Bi tangential Interpolation condition2}) with respect to $s$ and  evaluating at $s=\sigma$ we get (\ref{Bi-Tangential ErrorQ}). 
\end{proof}

Theorem 2 quantifies the interpolation errors based on which we can predict the performance of a limited time interpolation based model reduction scheme. The right tangential error depends on the interpolation point and the final time($\tau$). If the real part of the interpolation point is sufficiently positive or if $\tau \ll $, the right interpolation error will be  minor.  A lesser distance of the vector $(e^{A\Pi \tau}-e^{A\tau})Bb$ from the Kernel of $W_r^{*}$ also reduces the error. 

Similar to the previous case, an interpolation point with an adequately positive real part or sufficiently small final time $\tau$ ensures a slight left tangential error. In addition, the closeness of the vector $(e^{A^{*}{\Pi}^{*}\tau}-e^{A^{*}\tau})C^{*}c$ to the Kernel of $V_r^{*}$ also ensures a slight right tangential error. 

For the bi-tangential error (\ref{Bi-Tangential ErrorP},\ref{Bi-Tangential ErrorQ}) to be small, we require either $ R_{P_1}(\sigma)$ and $R_{P_2}(\sigma)$ or $ R_{Q_1}(\sigma)$ and $R_{Q_2}(\sigma)$ to be small. The criteria which ensure a small right tangential error and a small left tangential error as discussed previously also ensures a smaller bi-tangential error.

\begin{corollary}\label{Corollary 1}
Let $\hat{G}$ be the reduced order model discussed in Theorem \ref{Theorem 2}. Then $\hat{G}_{\tau}(s)$ nearly tangentially interpolates $G_{\tau}(s)$ at each interpolation point $\{\sigma_1,\sigma_2,\hdots,\sigma_r\}$ in corresponding tangential directions $\{\hat{b}_1,\hat{b}_2,\hdots,\hat{b}_r\}$ and $\{\hat{c}_1,\hat{c}_2,\hdots,\hat{c}_r\}$.  
\end{corollary}

The interpolation points along with the right and left tangential directions and the additional conditions discussed in the previous theorem determine how well $\hat{G}_{\tau}(s)$ approximates $G_{\tau}(s)$. 

\section{\textit{Limited Time IRKA}}\label{Section 4}
 Theorem 1 asserts that the time-limited optimality conditions will be satisfied provided the interpolation points determining the reduced order model tally with the reflected poles of the reduced model. The difficulty in constructing such reduced order models is that one doesn't know how to choose such interpolation data apriori. 

Inspired by the IRKA algorithm in \citet{gugercin2008h_2} and time-limited rational interpolation discussed in the previous section, we propose the Limited Time Iterative Rational Krylov Algorithm (LT-IRKA). 

\begin{algorithm}[h]
   \caption{LT-IRKA}
   \KwIn{The system matrices: $A,B,C$ \newline Initial interpolation points: $\{\sigma_1,\hdots,\sigma_r \}$;\newline Initial tangential directions: $\tilde{B}=[b_1,\hdots,b_{r}]$ and $\tilde{C}=[c_1,\hdots,c_r]$}.
   \KwOut{The reduced matrices $\hat{A},\hat{B},\hat{C}$} 
   \While{(\textrm{relative change in} $\{\sigma_i\}>$\textrm{tol})}{
    1.  Compute $V_r$ and $W_r$ using (\ref{RightProjectionSpace}),(\ref{LeftProjectionSpace}),(\ref{RightProjector}) and (\ref{LeftProjector}) respectively\;
    2.  Update ROM: 
    \begin{equation}\label{Reduced Matrices}
        \begin{aligned}
        \hat{A}&=(W_r^{*}V_r)^{-1}{W_r}^{*}AV_r \\
        \hat{B}&=(W_r^{*}V_r)^{-1}{W_r}^{*}B \\
        \hat{C}&=CV_r 
        \end{aligned}
    \end{equation}\;
    3. $\sigma_i=-\lambda_i(\Lambda)$, $\hat{A}=R\Lambda R^{-1}$, $\tilde{B}={B_r}^{*}R^{-*} $
    \newline and $\tilde{C}=C_r R$  \;
}
\end{algorithm}

 LT-IRKA produces high-fidelity reduced order models for the time-limited $H_2$ norm. Also, similar to TL-TSIA, the reduced model satisfies the $H_2(\tau)$ optimality conditions approximately. The nearness of the reduced order model to the actual optimality conditions can be gauged using the requirements given by Theorem \ref{Theorem 2}. 

\textit{Initialization:} 
We can choose the initial interpolation points and tangential directions for LT-IRKA using various techniques. Randomly selecting the initial interpolation points and tangential directions, initializing conventional IRKA randomly and using the reduced order system to obtain the initial interpolation points and tangential directions, effectively computing eigenvalues and corresponding left and right eigenvectors corresponding to the dominant residues using the dominant pole algorithm \citep{rommes2006efficient} etc.\ are some of the possible ways for initializing LT-IRKA.

\textit{Algorithm Implementation:}
For practical purposes, we require the reduced order equivalent of a real system to be real. This is possible if the complex interpolation points (and associated tangential directions) are chosen in conjugate pairs. This is a standard practice in Krylov subspace methods wherein the interpolation data is grouped into conjugate pairs to obtain a real basis satisfying (\ref{LeftProjectionSpace}) or (\ref{RightProjectionSpace}) \citep{grimme1997krylov}. The reduced order model is obtained using Petrov-Galerkin Projection. The reflection of the poles and the residues of the reduced order model are updated as the interpolation points and tangential directions for the next iteration, projectors are constructed and the reduced order model is obtained. This is continued till the convergence criteria is satisfied.

\textit{Convergence:}
 We choose the norm of the difference of eigenvalues of the reduced state matrix $\hat{A}$ for consecutive iterations as the convergence criteria for LT-IRKA, which is similar to IRKA. Though formal proof of the convergence of LT-IRKA is missing, we have tested the LT-IRKA algorithm on various examples. The algorithm converges after a finite number of iterations for a proper set of initial interpolation points and tangential directions.
 
\begin{remark}
The main difficulty in implementing the LT-IRKA algorithm involves computing the quantity $e^{A\tau}b$. Various methods are proposed in \cite{kurschner2018balanced} to handle this issue in a computationally efficient manner. In our work, we use the MATLAB routine \textit{expm} to compute  $e^{A\tau}$ and multiply it with the vector $b$ similar to \citet{goyal2019time}. This method works well for systems of order upto 1000. Apart from this, LT-IRKA  involves the use of matrix-vector multiplications and some linear solvers as in other Krylov based model reduction strategies. 
\end{remark}

\subsection{ Comparison with TL-TSIA: }\label{Subsection 4_1} 
The TL-TSIA is a near $H_2(\tau)$-optimal model order reduction algorithm which deals with Lyapuonv based $H_2(\tau)$ optimality conditions. Consider the system (\ref{eq1}) with state matrices $(A,B,C)$. The initial reduced order system (\ref{eq2}) with state matrices $(\hat{A}, \hat{B}, \hat{C})$ is obtained by using IRKA to reduce system (\ref{eq1}). Thereafter, right and left projectors are obtained by solving the following Sylvester equations. 
\begin{equation}\label{SylvestreEq1}
AP+P\hat{A}^{*} + B{\hat{B}}^{*}-e^{A\tau}B{\hat{B}}^{*}e^{{\hat{A}}^{*}\tau}=0 
\end{equation}
\begin{equation}\label{SylvesterEq2}
A^{*}Q+Q\hat{A}+C^{*}{\hat{C}}-e^{A^{*}\tau}C^{*}{\hat{C}}e^{\hat{A}\tau} =0 .
\end{equation} 
We select $P$ as the right projector. The left projector is obtained as $(Q^{*}P)^{-1}Q^{*}$ to ensure that $((Q^{*}P)^{-1}Q^{*})P=I_r$. The projectors are used to construct new reduced matrices $(\hat{A},\hat{B},\hat{C})$ using (\ref{Reduced Matrices}) which are used in the next iteration to construct a new set of projectors and the iterations are continued until the relative change in the eigenvalues of the reduced state matrix $\hat{A}$ become less than a fixed value. This is similar to the convergence criteria used in the algorithm LT-IRKA proposed in this work. Inspired by the similarity comparison between TSIA and IRKA in \citet{benner2011sparse} for SISO systems, we propose a theorem that shows that TL-TSIA and LT-IRKA produce similar projection subspaces and hence will result in identical reduced order models upon convergence.

\begin{theorem} \label{Theorem 3}
The system $\Sigma$ (\ref{eq1}) is a MIMO LTI system and $\hat{\Sigma}$ (\ref{eq2}) is the reduced order equivalent of $\Sigma$. We denote $\mathcal{V}$ and $\mathcal{W}$ as the final right and left projection spaces, respectively, obtained when LT-IRKA converges for $\Sigma$. $P$ and $Q$ are the final right and left projection matrices obtained using TL-TSIA. Assume that final time $\tau$ is kept constant, and both the algorithms converge. Then the following is true:
\begin{align*}
\mathcal{V} \equiv \textrm{span}(P)\\
\mathcal{W} \equiv \textrm{span}(Q)
\end{align*} 
\end{theorem}

\begin{proof} 
We shall first prove the equivalence of the right projection subspace, $\mathcal{V}$ (\ref{RightProjectionSpace}) to the right projection space spanned by the columns of the matrix $P$ obtained by solving the Sylvester equation (\ref{SylvestreEq1}). Let $\hat{A} = SDS^{-1}$ where $D$ is a diagonal matrix and the columns of matrix $S$ are the right eigenvectors. Then, $e^{\hat{A}\tau}$ becomes $S e^{D\tau} S^{-1}$. Substituting these expressions for $\hat{A}$ and $e^{\hat{A} \tau}$ in (\ref{SylvestreEq1}) we have,
\begin{align*}
 AP+P(SDS^{-1})^{*}+B{\hat{B}}^{*}-&\\
e^{A\tau}B{\hat{B}}^{*}(Se^{D\tau}S^{-1})^{*} &=0 \\
 AP+ PS^{-*}D^{*}S^{*}+B{\hat{B}}^{*}-&\\
 e^{A\tau}B{\hat{B}}^{*}S^{-*}e^{D^{*}\tau}S^{*} &=0 \\
 APS^{-*}+PS^{-*}D^{*}+B{\hat{B}}^{*}S^{-*}-&\\
 e^{A\tau}B{\hat{B}}^{*}S^{-*}e^{D^{*}\tau}&=0 
\end{align*}

\begin{equation}\label{TLIRKA1}
A\hat{P}+\hat{P}D^{*}+B{\hat{B}}^{*}-e^{A\tau}B{\hat{B}}^{*}e^{D^{*}\tau} =0
\end{equation} 

We denote $PS^{-*}$ as $\hat{P}$ and ${B_r}^{*}S^{-*}$ as $\tilde{B}$. The columns of $\hat{P}$ are represented as $\{ \hat{P}_i, i=1,2,\hdots,r\}$. Since $S$ is a non-singular matrix so the columns of $P$ and $\hat{P}$ span the same subspace i.e.\ $\textrm{span}(P)=\textrm{span}(\hat{P})$. Let us denote the columns of $\tilde{B}$ as $\{\hat{b}_i,i=1,2,\hdots,r \}$ and assume $D= \textrm{diag}\{\lambda_1,\lambda_2,\hdots,\lambda_r\}$. The negative of the eigenvalues are considered as the interpolation points, $-\lambda_i = \sigma_i $ and  the columns $\hat{b}_i$ happen to be the right tangential directions  of the LT-IRKA algorithm. The equation (\ref{TLIRKA1}) rewritten in terms of the $i^{th}$ column of $P$ becomes
\begin{equation}
A\hat{P}_i+\hat{P}_i{{\lambda}^{*}_i}+B\hat{b}_i-e^{A\tau}B\hat{b}_ie^{{\lambda^{*}_i}\tau}=0 
\end{equation}
The equation can be rewritten as,
\begin{align*}
(\sigma^{*}_i I_n-A)\hat{P}_i = (I_n-e^{-\sigma^{*}_i \tau}e^{A\tau})B\hat{b}_i \\
 \hat{P}_i = (\sigma^{*}_i I_n-A)^{-1}(I_n-e^{-\sigma^{*}_i \tau}e^{A\tau})B\hat{b}_i
\end{align*} 
Thus, the columns of $\hat{P}$  span the subspace $\textrm{span}\{(\sigma^{*}_i I_n-A)^{-1}(I_n-e^{-\sigma^{*}_i \tau}e^{A\tau})B\hat{b}_i, i= 1,\hdots,r \}$ which is the right projection subspace $\mathcal{V}$ of the TL-IRKA Algorithm. Since the matrix $S$ is invertible, the column space of $\hat{P}$ and $P$ are same. We have shown that $\mathcal{V}=\textrm{span}(P)$. We can similarly show that $\mathcal{W}= \textrm{span}(Q)$.
\end{proof}

This theorem shows that both TL-TSIA and LT-IRKA, upon convergence yield the same left and right Projection subspaces. For a fixed reduced order $r$, both algorithms yield similar $H_{2}(\tau)$ optimal reduced order model of the original system. This will be verified with numerical examples in the next section.

\section{Numerical Examples}\label{Section 5}
We investigate the performance of LT-IRKA and compare its efficiency with TL-TSIA, TL-BT, IRKA and TL-PORK using two single-input single-output (SISO) examples. For the third example, which is a multi-input multi-output model, the performance of LT-IRKA is compared with TL-TSIA, TL-BT and IRKA. The first example is a SISO clamped beam model of order 348.  The second example is a FOM model with single input single output(SISO) of order 1006. The third example is the International Space Station (ISS) model with three inputs and three outputs and order 270. We obtain the examples from http://slicot.org/20-site/126-benchmark examples-for-model-reduction. The simulations are done in MATLAB version 8.3.0.532(R2014a) on a Intel(R) Core(TM) i5-6500 CPU @ 3.20GHz 3.19 GHz system with 16 GB RAM. 

 We shall reduce the three models over various finite time intervals of the form $[0,\tau]$ for a fixed final time instant $\tau$. Using LT-IRKA, we obtain a lower order equivalent of the models and compute the absolute and relative $H_2(\tau)$ errors. We calculate the same errors for lower-order counterparts obtained using TL-BT, TL-TSIA, IRKA and TL-PORK. TL-PORK-1 and TL-PORK-2 correspond to initialization with the reduced system obtained from TL-IRKA and IRKA, respectively. LT-IRKA results in near $H_2(\tau)$-optimal reduced-order models. To quantify how close the reduced-order system comes to satisfying the $H_2(\tau)$ optimality conditions, we define the following quantities,

\begin{enumerate}

\item Right tangential error, $\text{RTerr}_{\text{rel}}$ is defined as 
\begin{equation}\label{relRTerror}
\begin{bmatrix} \frac{\left\Vert  G_{\tau}(\sigma_{1})b_1-\hat{G}_{\tau}(\sigma_{1})b_1 \right\Vert}{\left\Vert G_{\tau}(\sigma_{1})b_1 \right\Vert} \hdots \frac{\left\Vert G_{\tau}(\sigma_{r})b_r-\hat{G}_{\tau}(\sigma_{r})b_r \right\Vert}{\left\Vert G_{\tau}(\sigma_{r})b_r \right\Vert}\end{bmatrix}
\end{equation}

\item Left tangential error, $\text{LTerr}_{\text{rel}}$ is defined as
\begin{equation}\label{relLTerror}
\begin{bmatrix}
\frac{\left\Vert {c_1}^{*}G_{\tau}(\sigma_{1})-{c_1}^{*}\hat{G}_{\tau}(\sigma_{1})  \right\Vert} {\left\Vert {c_1}^{*}G_{\tau}(\sigma_{1}) \right\Vert}  \hdots    \frac{\left\Vert {c_r}^{*}G_{\tau}(\sigma_{r})-{c_r}^{*}\hat{G}_{\tau}(\sigma_{r})  \right\Vert} {\left\Vert {c_r}^{*}G_{\tau}(\sigma_{r}) \right\Vert} \end{bmatrix}
\end{equation}
 
\item Bi-tangential error, $\text{derr}_{\text{rel}}$ is defined as 
\begin{equation} \label{relderror}
\begin{bmatrix} \frac{\left\Vert {c_1}^{*}G'_{\tau}(\sigma_{1})b_1- {c_1}^{*}\hat{G}'_{\tau}(\sigma_{1})b_1 \right\Vert} {\left\Vert {c_1}^{*}G'_{\tau}(\sigma_{1})b_1 \right\Vert}  \hdots  \frac{\left\Vert {c_i}^{*}G'_{\tau}(\sigma_{i})b_i- {c_i}^{*}\hat{G}'_{\tau}(\sigma_{i})b_i \right\Vert} {\left\Vert {c_i}^{*}G'_{\tau}(\sigma_{i})b_i \right\Vert}\end{bmatrix} 
\end{equation}

\end{enumerate}
The first and second examples are SISO systems and hence the right and left tangential errors are the same. In this case, we represent the relative errors as $\text{RTerr}_{\text{rel}}$ = $\text{LTerr}_{\text{rel}}$= $\text{Err}_{\text{rel}}$
\begin{figure}[H]
\includegraphics[width = \textwidth]{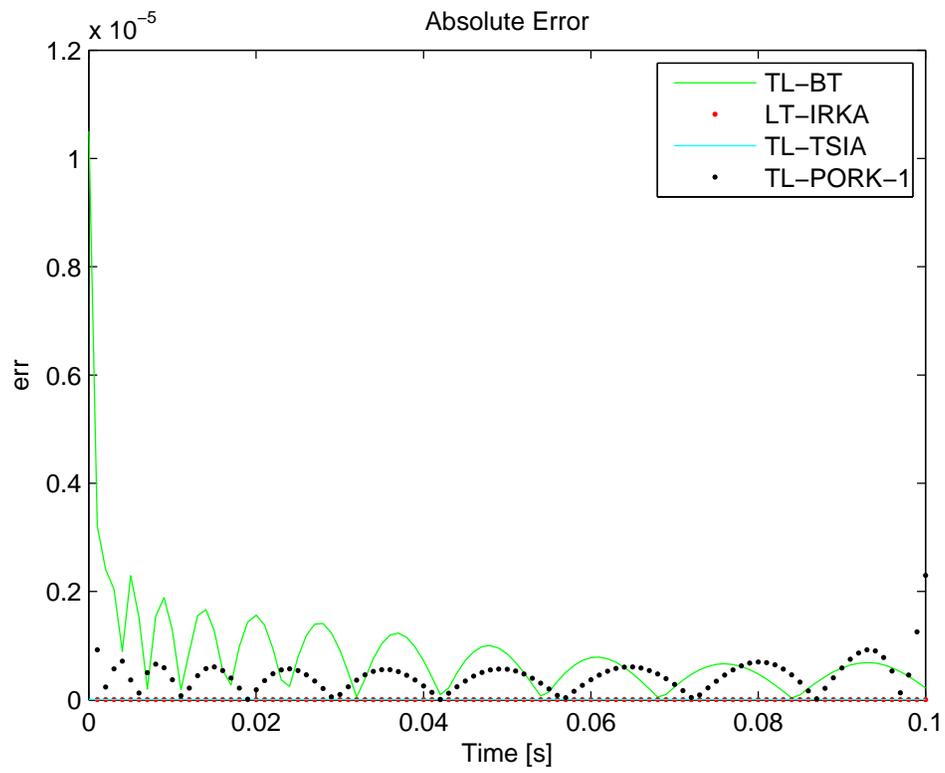}
\centering
\caption{Beam Example. Final time, $\tau$ = 0.1 s }
\label{Figure 1}
\end{figure}

\begin{figure}[H]
\includegraphics[width = \textwidth]{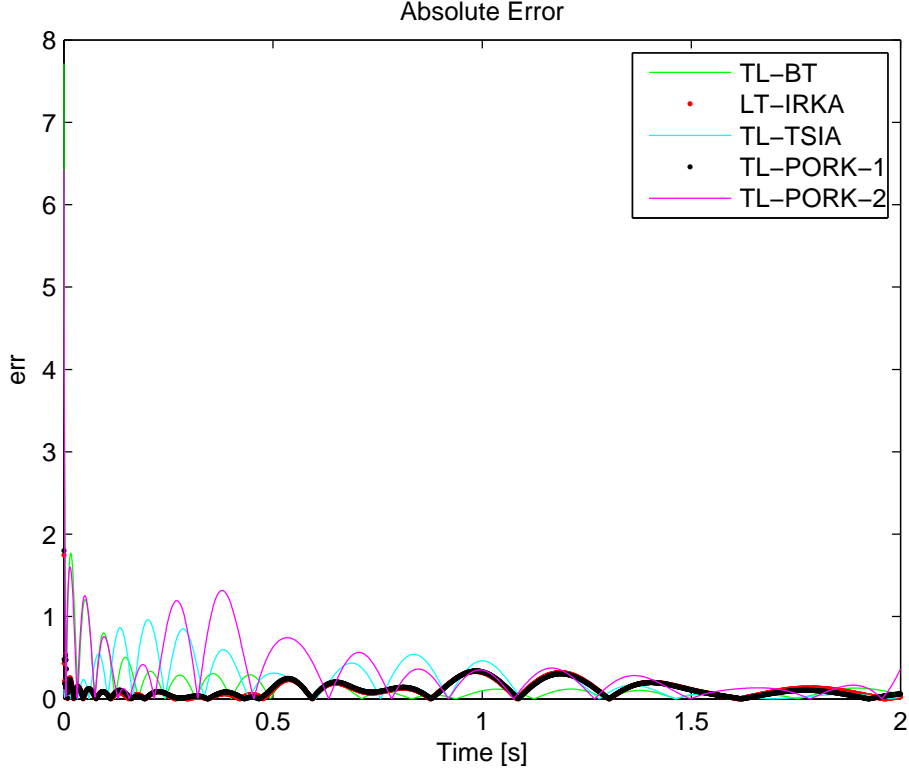}
\centering
\caption{Beam Example: $\tau$ = 2 s}
\label{Figure 2}
\end{figure}

\paragraph*{Example 1:}
The first numerical example is a beam model of order $n=348$. We obtain reduced-order equivalents with $r=12$ for two time intervals $\begin{bmatrix}0 & 0.1 \end{bmatrix}$ and $\begin{bmatrix}0 & 2 \end{bmatrix}$, using various algorithms. Fixing the error tolerance at $10^{-5}$, we initialize the interpolation points for LT-IRKA randomly. The LT-IRKA algorithm converges for both the time intervals. The number of iterations required for convergence depends on the initial interpolation points. The absolute and relative $H_2(\tau)$ errors for various reduction methods are compared in Table \ref{table1} and Table \ref{table2} for the two time intervals. Also, the absolute error response of the reduced models obtained through various algorithms is plotted in Figure \ref{Figure 1} and Figure \ref{Figure 2}. Consider the final time instant $\tau = 0.1$ s. From Table \ref{table1} we see that the relative $H_2(\tau)$ error of the reduced system obtained from LT-IRKA and TL-TSIA is several orders of magnitude less than the reduced models obtained by the other algorithms, including TL-BT, TL-PORK-1 and TL-PORK-2. In case final time instant $\tau=2$ s,  LT-IRKA and TL-PORK-1 perform better than the other algorithms concerning relative $H_2(\tau)$ error as evident from Table \ref{table2}. Referring to Figure \ref{Figure 1} and \ref{Figure 2}, we observe that LT-IRKA yields a high fidelity reduced order system for both the time intervals under consideration.

\begin{table}[H]
\tbl{Relative $H_2(\tau)$ Errors for Beam example.}
{\begin{tabular}{p{1.8cm}cccccc} \toprule
Algorithm  & TL-BT & LT-IRKA & IRKA & TL-TSIA & TL-PORK-1 & TL-PORK-2 \\ \midrule
 Rel$\left\Vert \textrm{Err} \right\Vert_{H_2(\tau)}$ $\tau = 0.1$ s & $6.79 \times $ $10^{-8}$ & $6.55 \times 10^{-11}$ & $0.0580$ & $6.85 \times 10^{-11}$ & $9.25\times 10^{-9}$ & $0.0123$ \\ \midrule

 Rel$\left\Vert \textrm{Err} \right\Vert_{H_2(\tau)}$ $\tau = 2$ s  & $0.0262$ & $0.0115 $ & $0.0476$ & $0.0243$ & $0.0114$ & $0.0375$ \\ \bottomrule
\end{tabular}}
\label{table1}
\end{table}

\begin{table}[H]
\tbl{Relative error in the optimality conditions for Beam example}
{\begin{tabular}{lccc} \toprule
 
 Final-Time & Algorithm & $\left\Vert \textrm{Err} \right\Vert_{\textrm{rel}}$ & $\left\Vert \textrm{dErr} \right\Vert_{\textrm{rel}}$\\ \midrule
\multirow{2}{*}{$\tau=0.1$ s} &  LT-IRKA & $4.48 \times 10^{-12}$ & $1.32 \times 10^{-11}$ \\ 
 &  IRKA & $0.0045$ & $5.6221$ \\ \midrule
\multirow{2}{*}{$\tau=2$ s} &  LT-IRKA & $ 0.0028 $ & $ 0.0187 $\\ 
 &  IRKA & $ 0.0804 $ & $ 8.3027 $ \\ \bottomrule
\end{tabular}}
\label{table2}
\end{table}

The real component of the interpolation points, relative right and left tangential errors and relative bi-tangential errors (\ref{relRTerror}, \ref{relLTerror} and \ref{relderror}) for the reduced model obtained by LT-IRKA for the time interval $\begin{bmatrix}0 & 0.1 \end{bmatrix}$ are given below. 
\begin{align*}
\textrm{Re}(\sigma) &= [527 ,454.7,138.1,138.1,73.3,73.3,28.2,\\
& \quad 28.2,13.2,13.2,6,-12 ],\\
\textrm{Err}_{\text{rel}} &= 10^{-13} \times [0.4,0.1,0.2,0.2,1.4,1.4,102.5,\\ 
& \quad 102.5, 172.4, 172.4, 125.1, 327.2],\\
\quad \textrm{derr}_{\text{rel}} &=  10^{-13} \times [0.4,0.2,0.2,0.2,9.2 ,9.2,736.3,\\ & \quad 736.3,404.0,404.0,212.5,628.6]
\end{align*} 

The same information for the time interval $\begin{bmatrix}0 & 2 \end{bmatrix}$ are as follows,
\begin{align*}
\textrm{Re}(\sigma) &= [45.6,45.6,11.1,11.1,3.6,3.6,0.6,0.6,\\
& \quad -0.1,0.7,0.9,0.9 ],\\
\textrm{Err}_{\text{rel}} &= [0.1\times 10^{-14},0.1\times 10^{-14},3.5 \times 10^{-13}, \\
& \quad 3.5 \times 10^{-13},2.4 \times 10^{-6}, 2.4 \times 10^{-6}, \\  & \quad  0.0017,
       0.0005, 0.0002, 0.0009, 0.0009],\\
\textrm{derr}_{\text{rel}} &=  [0.6 \times 10^{-11},0.6 \times 10^{-11},0.1 \times 10^{-9},\\ & \quad 0.1 \times 10^{-9}, 0.8\times 10^{-4},0.8 \times 10^{-4},0.008,\\ & \quad 0.008,0.0009,0.0005,0.0105,0.0105]
\end{align*} 
 The interpolation points are the reflection of the eigenvalues of the reduced system. We observe from the above information that interpolation points corresponding to sufficiently negative eigenvalues result in minimal interpolation errors and almost converge to the optimality conditions. Interpolation points related to eigenvalues lying close to the imaginary axis or in the right hand side (RHS) of the s-plane yields high interpolation errors and causes the reduced model to deviate from the optimality conditions. 
Finally, we observe from Table 2 that LT-IRKA yields better near optimal $H_2(\tau)$ models than IRKA. Also, relative interpolation errors are negligible for the short time interval ($\tau = 0.1$ s). In comparison, the errors are several orders of magnitude higher for the second time interval ($\tau = 2$ s).  These observations validate the dependence of the interpolation error on the individual interpolation points and the final time instant ($\tau$) as discussed in Theorem \ref{Theorem 2}. 
\begin{figure}[H]
\includegraphics[width = \textwidth]{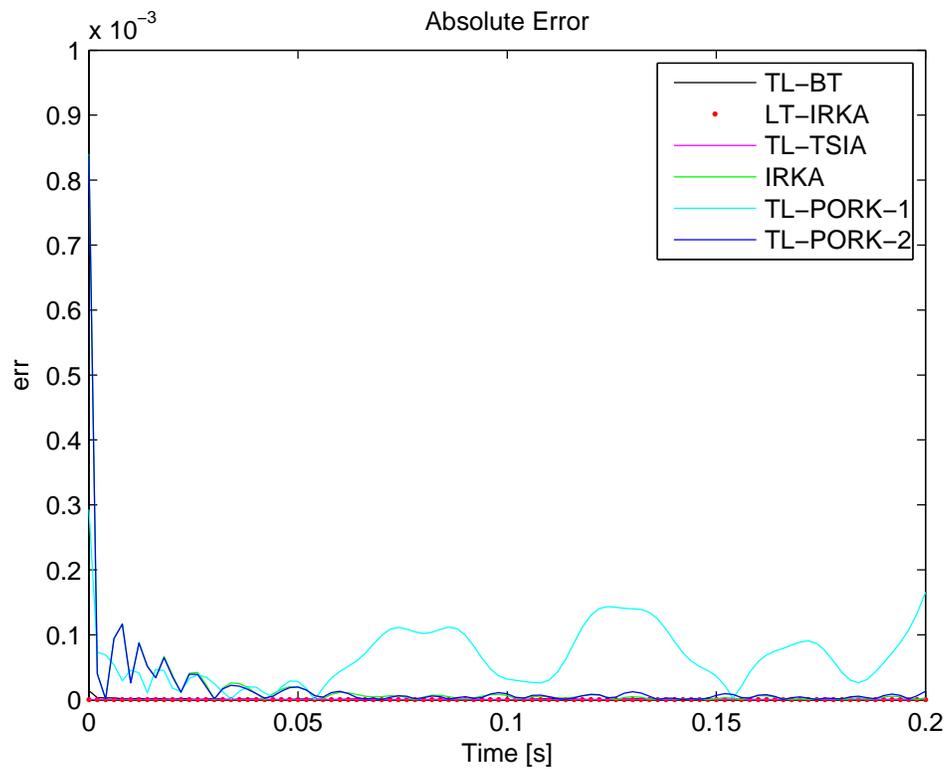}
\centering
\caption{FOM Example: $\tau$ = 0.2 s}
\label{Figure 3}
\end{figure}

\begin{figure}[H]
\includegraphics[width = \textwidth]{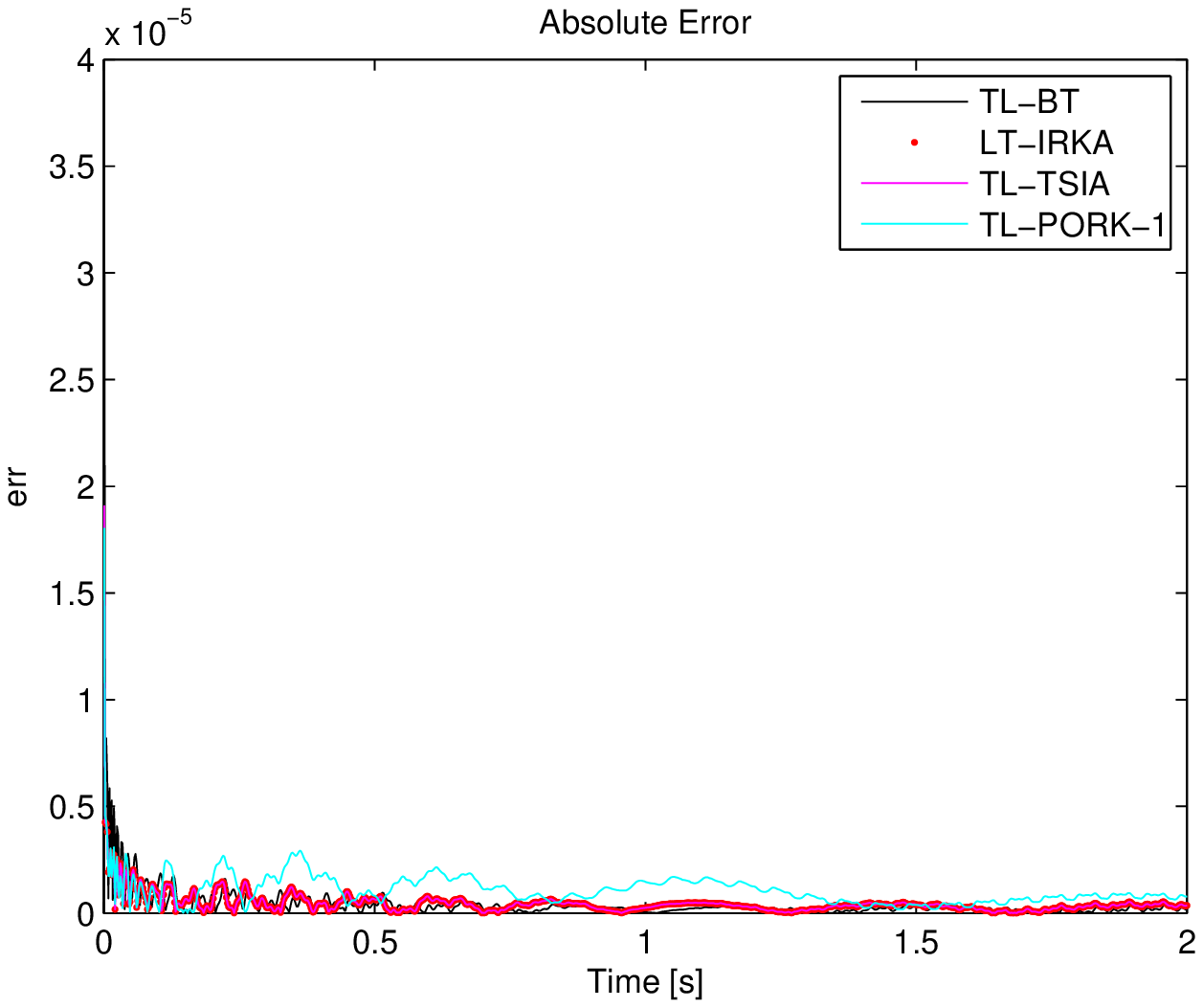}
\centering
\caption{FOM Example: $\tau$ = 2 s}
\label{Figure 4}
\end{figure}

\begin{table}[H]
\tbl{Relative $H_2(\tau)$ Errors in FOM example.}
{\begin{tabular}{p{1.8cm}cccccc} \toprule
Algorithm  & TL-BT & LT-IRKA & TL-TSIA & IRKA & TL-PORK-1 & TL-PORK-2\\ \midrule
 Rel$\left\Vert \textrm{Err} \right\Vert_{H_2(\tau)}$ $\tau = 0.2$ s & $5.79 \times 10^{-9} $ & $5.59 \times 10^{-12}$  & $5.59 \times 10^{-12} $ & $3.01 \times 10^{-7} $ & $3.14 \times 10^{-7}$ & $3.01 \times 10^{-7}$ \\ \midrule

 Rel$\left\Vert \textrm{Err} \right\Vert_{H_2(\tau)}$ $\tau = 2$ s  & $1.06 \times 10^{-8} $ & $6.31 \times 10^{-9}$  & $6.31 \times 10^{-9} $ & $2.05 \times 10^{-7} $ & $1.05 \times 10^{-8}$ & $2.05 \times 10^{-7}$ \\ \bottomrule
\end{tabular}}
\label{table3}
\end{table}

\begin{table}[H]
\tbl{Relative error in the optimality conditions for FOM example}
{\begin{tabular}{lccc} \toprule
 
 Final-Time & Algorithm & $\left\Vert \textrm{Err} \right\Vert_{\textrm{rel}}$ & $\left\Vert \textrm{dErr} \right\Vert_{\textrm{rel}}$\\ \midrule
\multirow{2}{*}{$\tau=0.2$ s} &  LT-IRKA &  $2.03 \times 10^{-12}$ & $1.23 \times 10^{-10}$  \\ 
 &  IRKA & $2.60 \times 10^{-8}$ & $1.02 \times 10^{-5}$  \\ \midrule
\multirow{2}{*}{$\tau=2$ s} &  LT-IRKA & $4.24 \times 10^{-10} $ & $1.35 \times 10^{-8}$ \\ 
 &  IRKA &  $1.07 \times 10^{-9}$ & $3.20 \times 10^{-5}$ \\ \bottomrule
\end{tabular}}
\label{table4}
\end{table}

\paragraph*{Example 2:}
The second example is a FOM model of order 1006. Reduced models of order $r=20$ are obtained using LT-IRKA,TL-BT, TL-TSIA, IRKA, TL-PORK-1 and TL-PORK-2. The error tolerance is fixed at $10^{-5}$. We randomly initialize the interpolation points for LT-IRKA. For the time interval $\begin{bmatrix}0 & 0.2 \end{bmatrix}$, we compare the relative $H_2(\tau)$-errors in Table \ref{table3}. LT-IRKA performs better than TL-BT, IRKA, TL-PORK-1 and TL-PORK-2 and yields the same relative $H_2(\tau)$ error as TL-TSIA. Next, we consider a larger time interval $\begin{bmatrix}0 & 2 \end{bmatrix}$. Table \ref{table3} compares the relative $H_2(\tau)$ errors for this time-interval. Again, LT-IRKA has a lesser $H_2(\tau)$ error than TL-BT, IRKA, TL-PORK-1 and TL-PORK-2 and similar $H_2(\tau)$ error as TL-TSIA. Figures \ref{Figure 3} and \ref{Figure 4} show that LT-IRKA approximates the original system extremely well for both time intervals. Table \ref{table4} compares the relative interpolation errors of LT-IRKA and IRKA for both time-intervals. LT-IRKA performs better than IRKA for both time intervals. 
\begin{figure}[H]
\includegraphics[width = \textwidth]{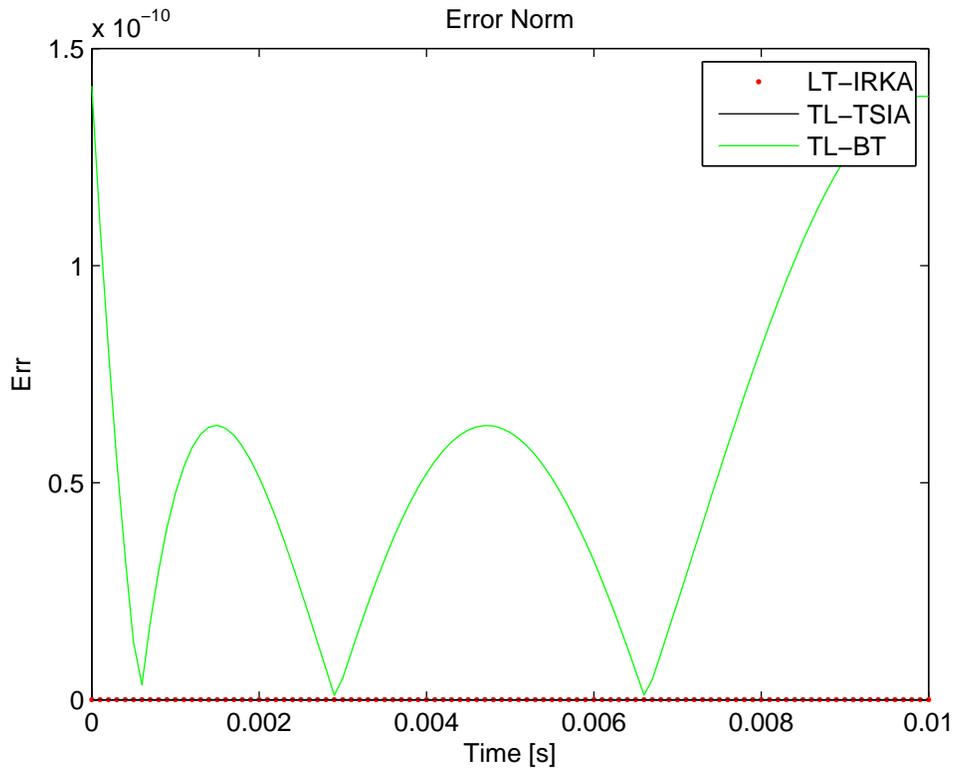}
\centering
\caption{ISS Example: $\tau$ = 0.01 s}
\label{Figure 5}
\end{figure}

\begin{figure}[H]
\includegraphics[width = \textwidth]{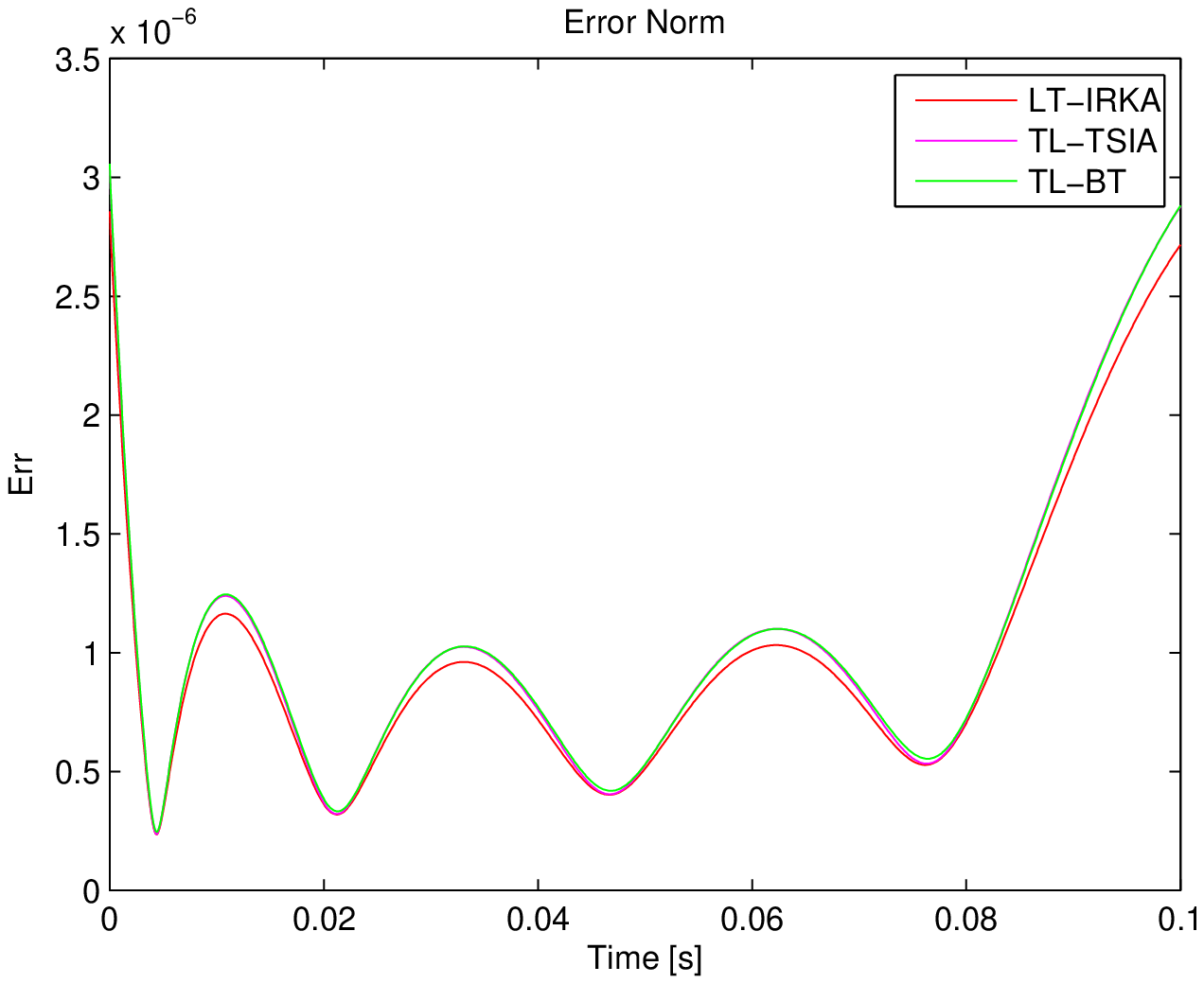}
\centering
\caption{ISS Example: $\tau$ = 0.1 s}
\label{Figure 6}
\end{figure}

\begin{figure}[H]
\includegraphics[width = \textwidth]{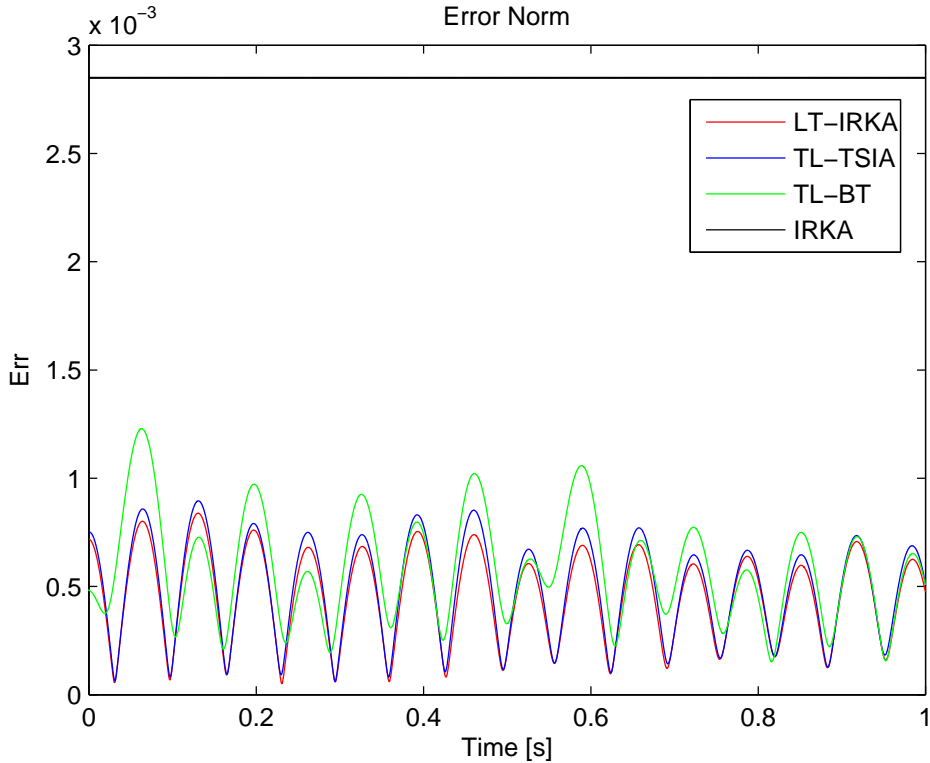}
\centering
\caption{ISS Example: $\tau$ = 1 s}
\label{Figure 7}
\end{figure}

\paragraph*{Example 3:}
The final example discussed is a model of the International Space Station(ISS). Reduced order equivalents for $r=12$ are obtained by applying LT-BT, LT-IRKA, LT-TSIA and IRKA over three time intervals $\begin{bmatrix} 0 & 0.01\end{bmatrix}$, $\begin{bmatrix} 0 & 0.1\end{bmatrix}$  and $\begin{bmatrix} 0 & 1\end{bmatrix}$. The error tolerance is fixed at $10^{-8}$. The initial interpolation points and tangential directions are randomly chosen, and LT-IRKA converges for the three time intervals. The $H_2(\tau)$ errors of all the reduced order models are compared in Table \ref{table5}. For the smallest time, $\tau=0.01$, LT-IRKA and TL-TSIA perform better than TL-BT and IRKA. For the other two final time instants, $\tau=0.1$ and $\tau=1$, LT-IRKA yields similar $H_2(\tau)$ error as TL-BT and TL-TSIA and lesser $H_2(\tau)$ error compared to IRKA. Since the system is MIMO, we plot the error norm trajectory in Figure \ref{Figure 5}, \ref{Figure 6}  and \ref{Figure 7} for the three final time instants. The error in the interpolation based $H_2(\tau)$ optimality conditions for the three time intervals are given in Table \ref{table6}. For every time interval, LT-IRKA performs better than IRKA. However, for the smallest time interval, LT-IRKA performs better than IRKA by several orders of magnitude.

\begin{table}[H]
\tbl{Relative $H_2(\tau)$ Errors in FOM example.}
{\begin{tabular}{p{1.8cm}cccc} \toprule
Algorithm  & TL-BT & LT-IRKA & TL-TSIA & IRKA\\ \midrule
 Rel$\left\Vert \textrm{Err} \right\Vert_{H_2(\tau)}$ $\tau = 0.01$ s & $9.84 \times 10^{-9} $ & $2.0319 \times 10^{-12}$  & $2.0332 \times 10^{-12} $ & $0.6940 $ \\ \midrule

 Rel$\left\Vert \textrm{Err} \right\Vert_{H_2(\tau)}$ $\tau = 0.1$ s  & $2.99 \times 10^{-4}$ & $2.9962 \times 10^{-4} $  & $2.9923 \times 10^{-4}$ & $0.8657 $  \\ \midrule
 
 Rel$\left\Vert \textrm{Err} \right\Vert_{H_2(\tau)}$ $\tau = 1$ s  & $0.1946$ & $0.1685 $  & $0.1684  $ & $0.8774$ \\  \bottomrule
\end{tabular}}
\label{table5}
\end{table}

\begin{table}[H]
\tbl{Relative error in the optimality conditions for FOM example}
{\begin{tabular}{lcccc} \toprule
 
 Final-Time & Algorithm & $\left\Vert \textrm{RTErr} \right\Vert_{\textrm{rel}}$ &  $\left\Vert \textrm{LTErr} \right\Vert_{\textrm{rel}}$ & $\left\Vert \textrm{dErr} \right\Vert_{\textrm{rel}}$ \\ \midrule
\multirow{2}{*}{$\tau=0.01$ s} &  LT-IRKA & $1.84 \times 10^{-12}$ & $1.83 \times 10^{-12}$ & $1.94 \times 10^{-9}$ \\ 
 &  IRKA & $2.2988$ & $2.2991 $ & $2.2879$ \\ \midrule
\multirow{2}{*}{$\tau=0.1$ s} &  LT-IRKA & $7.68 \times 10^{-4}$ & $8.69 \times 10^{-4}$ & $0.0021 $  \\ 
 &  IRKA & $1.5786$ &   $1.5451$ & $9.4650$  \\ \midrule
\multirow{2}{*}{$\tau=1$ s} &  LT-IRKA & $0.0575$ & $0.0608$ & $0.1398$ \\ 
 &  IRKA & $0.3206$ & $0.4175$ & $0.4214$  \\ 
 \bottomrule
\end{tabular}}
\label{table6}
\end{table}

For the three examples considered, LT-IRKA yields high fidelity reduced order models for various time intervals considered and has certain advantages. It outperforms TL-BT, an important time-limited model reduction algorithm for the smaller time intervals. For the case of TL-TSIA, both the algorithms yield comparable results as predicted in Theorem \ref{Theorem 3}. However, the reduced system obtained by initializing IRKA with random initial conditions is used as an initial guess for TL-TSIA \citep{goyal2019time} whereas LT-IRKA is initialized randomly. Also, LT-IRKA performs better than TL-PORK-1 and TL-PORK-2, as evident from Table \ref{table1} and Table \ref{table3}. Finally, from Table \ref{table2}, Table \ref{table4} and Table \ref{table6} we can see that the reduced models obtained by LT-IRKA satisfy the $H_2(\tau)$-optimality conditions with a higher degree of accuracy compared to the reduced order models obtained from IRKA for all time intervals under consideration.

\section{Conclusion}\label{Section 6}
In this paper, we derive interpolation-based Meier-Luenberger $H_2(\tau)$ optimality conditions. We propose a near $H_2(\tau)$ optimal interpolation-based algorithm called LT-IRKA for model order reduction of LTI systems over a time interval of the form $\begin{bmatrix}0 & \tau \end{bmatrix}$. The algorithm yields a near $H_2(\tau)$-optimal reduced order system. The errors in the right tangential, left tangential and bi- tangential interpolation conditions for $H_2(\tau)$-optimality are quantified. We compare LT-IRKA with another near $H_2(\tau)$ optimal reduction algorithm called TL-TSIA. We apply the LT-IRKA algorithm to reduce three LTI models for various time intervals and compare its performance with various time-limited model reduction algorithms and IRKA. The numerical simulations demonstrate the good performance of LT-IRKA compared to the considered model reduction algorithms over the time interval of interest.

\bibliographystyle{apacite}
\bibliography{references}

\end{document}